\documentclass[10pt]{article}
\usepackage{latexsym}
\usepackage{amssymb}
\usepackage{amsmath}
\usepackage{amsthm}
\usepackage[dvips]{graphicx}

\theoremstyle{theorem} \newtheorem{lemma}{Lemma}
\theoremstyle{theorem} \newtheorem{theorem}{Theorem}
\theoremstyle{theorem} \newtheorem{corollary}{Corollary}

\title{Gauge invariant plane-wave solutions in supersymmetric Yang-Mills quantum mechanics}
\author{Piotr Korcyl\thanks{e-mail address: korcyl@th.if.uj.edu.pl} \\ \small{\emph{M. Smoluchowski Institute of Physics, Jagiellonian
University}} \\ \small{\emph{Reymonta 4, 30-059 Krak\'{o}w,
Poland}}}
\date{\today}

\begin{document}

\maketitle

\begin{abstract}
We derive the spectra of $D=2$, $SU(3)$ supersymmetric Yang-Mills
quantum mechanics in all fermionic sectors. Moreover, we provide exact expressions for the corresponding eigenvectors
in the sectors with none and one fermionic quantum. % by solving a recurrence relation on the components of energy eigenstates in the Fock basis.
We also generalize our results obtained in a cut Fock space to the infinite cut-off limit.
\end{abstract}

\section{Motivations}

Gauge theories are believed to describe all forces in Nature.
One of the means to investigate the main characteristics of these theories is the study of their reduced to one point in
space versions, called  Yang-Mills quantum mechanics (YMQM) \cite{luscher1}\cite{luscher2}\cite{baal1}\cite{baal2}. From this viewpoint, one can calculate the
zero-volume limit of the spectrum of the theory, then find small volume corrections if needed.
Such approach yields the possibility to analyze the gauge dynamics which,
even in the case of the simplest, $D=2$, of these reduced systems,
turns out to be non-trivial due to the singlet constraint.
If one is interested in supersymmetry, it can be incorporated in the Yang-Mills quantum mechanics as well,
giving rise to supersymmetric Yang-Mills quantum mechanics (SYMQM).

A different motivation for studying such systems comes from the work by Hoppe \cite{hoppe} and de Witt, Hoppe and Nicolai \cite{dewit},
who showed that YMQM (SYM\-QM) describe the regularized dynamics of a relativistic quantum membrane (supermembrane).
Hence, the eigenstates of YMQM (SYMQM) turn out to be the wavefunctions of quantum membrane (supermembrane).
However, SYMQM in higher dimensional spaces, where the definition of a supermembrane
is consistent, are interacting theories, and difficult to solve. Thus, although the $D=2$ case is unphysical,
an analytic expression for its eigenstates is of interest,
since it may provide some indications on the form of higher dimensional solutions.

Several years ago a program to investigate numerically the whole class of supersymmetric quantum mechanical systems in a Hamiltonian formulation
by a cut Fock space method was proposed by Wosiek \cite{wosiek1}.
Numerous systems have already been studied and much analytical insight was inspired by these numerical
results \cite{wosiek3}\cite{wosiek4}\cite{wosiek5}\cite{kotanski2}\cite{doktorat_macka}\cite{korcyl_phd}(and the references therein).

Although the simplest case of $D=2$ supersymmetric Yang-Mills quantum mechanics with the $SU(2)$ symmetry group
was solved by Claudson and Halpern long time ago \cite{claudson}, very few exact solutions were found for other
groups. Some solutions were obtained in the bosonic sector \cite{doktorat_macka}\cite{samuel} for a general $SU(N)$ group, however
the fermionic eigenvalues and eigenvectors seem to be unknown in the literature.
In this work we present an analytic method inspired by the numerical approach of Wosiek to derive
exact solutions to the $D=2$ SYMQM valid in all fermionic sectors.
We demonstrate our approach by solving the nontrivial system with $SU(3)$ symmetry group.
We obtain the spectra in all fermionic sectors and, moreover, find closed expressions for the eigenvectors
in the bosonic sector and in the sector with one fermionic quantum. Generalizations of the eigenvectors to sectors with more fermions
are possible.
%Such problem
%was already investigated by Trzetrzelewski \cite{doktorat_macka} for the bosonic sector with
%a general $SU(N)$ symmetry using similar methods.

This paper is composed as follows. We start by introducing the class of supersymmetric Yang-Mills quantum mechanics
and presenting the particular model of $D=2$ with the $SU(3)$ gauge group. This is followed by a discussion of the cut Fock space
approach in section \ref{sec.The cut Fock basis method}. The derivation of the solutions is presented in subsequent sections:
in section \ref{sec.bosonic}, the bosonic sector is treated
and the closed expressions for the spectrum and eigenstates are calculated; in section \ref{sec. Solutions in the $n_F=1$ sector}
the spectrum as well as the eigenstates in the sector with one fermion are derived.
For both derivations solutions of recursion relations on the components
of energy eigenstate in the Fock basis is needed. The theorems used to solve these recurrences are presented collectively in the appendices.
Eventually, section \ref{sec. Conclusions} contains a summary of results.

\section{$D=2$, Supersymmetric Yang-Mills Quantum Mechanics}
\label{sec. $D=2$, Supersymmetric Yang-Mills Quantum Mechanics}

Supersymmetric Yang-Mills quantum mechanics can be obtained
by a dimensional reduction of a supersymmetric, $D=d+1$ dimensional Yang-Mills quantum field theory to one point in $d$-dimensional space.
The initial local gauge symmetry of field theory is thus reduced to a global symmetry of the quantum mechanical system.
 The simplest ones among such systems are those obtained by reduction of the
$\mathcal{N}=1$ Yang-Mills gauge theory in two dimensions \cite{claudson}. They are described by a scalar field $\phi_A$
and a complex fermion $\lambda_A$, where $A$ labels the generators of the gauge group. The Hamiltonian is\footnote{A sum over repeated indices is assumed.}
\begin{equation}
H = \frac{1}{2} \pi_A \pi_A + i g f_{ABC} \bar{\lambda}_A \phi_B \lambda_C,
\end{equation}
with the algebra of operators given by
\begin{equation}
[\phi_A, \pi_B] = i \delta_{A,B}, \qquad \{ \lambda_A, \bar{\lambda}_B \} = \delta_{A,B}.
\end{equation}
The generator of the gauge transformations is
\begin{equation}
G_A = f_{ABC}\big( \phi_B \pi_C - i \bar{\lambda}_B \lambda_C \big).
\end{equation}
The supersymmetry charges are
\begin{equation}
Q = \lambda_A \pi_A, \qquad \bar{Q} = \bar{\lambda}_A \pi_A,
\end{equation}
and
\begin{equation}
\{ Q, \bar{Q} \} = \pi_A \pi_A = 2H - 2 g \phi_A G_A.
\end{equation}
The fermionic term of the Hamiltonian is proportional to the generator of the gauge symmetry and hence it is supposed
to vanish on any physical (gauge invariant) state due to the Gauss law.
Hence, the Hamiltonian is simply
\begin{equation}
%H = \frac{1}{2} \sum_{A=1}^{N^2-1} \pi_A \pi_A.
H = \frac{1}{2} \pi_A \pi_A.
\label{eq.hamiltonian}
\end{equation}
Subsequently, we introduce bosonic creation and annihilation operators,
\begin{equation}
a_A = \frac{1}{\sqrt{2}}\big( \phi_A + i \pi_A\big), \qquad a^{\dagger}_A = \frac{1}{\sqrt{2}}\big( \phi_A - i \pi_A\big),
\end{equation}
and express them, as well as the fermionic creation and annihilation operators, in the matrix notation,
\begin{equation}
a_{i j} =  a_A T^A_{i j}, \qquad a^{\dagger}_{i j} = a^{\dagger}_A T^A_{i j}, \qquad
f_{i j} =  f_A T^A_{i j}, \qquad f^{\dagger}_{i j} = f^{\dagger}_A T^A_{i j},
\end{equation}
where $T^A_{i j}$ are the generators of the $SU(N)$ group in the fundamental representation.
Hence, all operators become operator valued $N \times N$ matrices.
Such a notation has a very practical feature, namely, any gauge invariant operator
can be simply written as a trace of a product of appropriate operator valued matrices \cite{maciek3}.
In situations when it is self-evident we will use a simplified notation for the trace of any such matrix, namely,
$\textrm{tr}(O) \equiv (O)$.
We get
\begin{equation}
Q^{\dagger} = \frac{i}{\sqrt{2}}\big( \textrm{tr} \ (f^{\dagger} a^{\dagger}) -  \textrm{tr} \ (f^{\dagger} a) \big), \quad \textrm{and} \quad Q = \frac{i}{\sqrt{2}}\big(  \textrm{tr} \ (f a^{\dagger}) -  \textrm{tr} \ (f a) \big),
\label{eq.supercharges}
\end{equation}
and
\begin{equation}
H =  \textrm{tr} \ (a^{\dagger}a) + \frac{N^2-1}{4} - \frac{1}{2}  \textrm{tr} \  (a^{\dagger} a^{\dagger}) - \frac{1}{2}  \textrm{tr} \ (aa).
\label{eq.hamiltonian_general}
\end{equation}

Therefore, in the case of the $SU(3)$ group we have,
\begin{equation}
H = \textrm{tr} \ (a^{\dagger}a) + 2 -\frac{1}{2} \textrm{tr} \ (a^{\dagger} a^{\dagger}) - \frac{1}{2}\textrm{tr} \ (aa).
\label{eq.hamiltonian2}
\end{equation}
Obviously, the Hamiltonian eq.\eqref{eq.hamiltonian2} conserves the fermionic occupation number. Hence, we can
analyze its spectrum separately in each subspace of the Hilbert space with fixed fermionic occupation number.
There are 9 such fermionic sectors because of the Pauli exclusion principle.
This Hamiltonian possess also another symmetry, namely, the particle-hole symmetry \cite{doktorat_macka}\cite{korcyl_phd},
which can be thought of as a quantum mechanical precursor of the charge conjugation symmetry in quantum field theory.
This symmetry can be observed as a perfect matching of the energy eigenvalues in the sectors with $n_F$
and $8-n_F$ fermionic quanta\footnote{This symmetry is present in $D=2$, SYMQM models with $SU(N)$ gauge group for any $N$.}.
Thus, it is sufficient to
consider only the sectors with $n_F=0, \dots, 4$ fermionic quanta.

\section{The cut Fock basis method}
\label{sec.The cut Fock basis method}

The cut Fock space approach turned out to be very useful in the studies of gauge
systems for several reasons.
First of all, it is a fully non-perturbative tool. Second, it treats
bosons and fermions on an equal footing, therefore calculations can be extended
to all fermionic sectors without difficulties.
Finally, it can be generalized, to handle $SU(N)$ gauge groups with $N\ge 2$,
as well as systems defined in spaces of various dimensionality.

The analytic approach used in this work is inspired by this numerical treatment. Hence, we now briefly summarize
the numerical approach in order to introduce the
basic notions needed in the following parts. For a more extensive discussion
see \cite{wosiek1}\cite{wosiek5}\cite{korcyl_phd}\cite{korcyl0}\cite{korcyl1}.

The basic ingredient of the numerical approach is a
systematic construction of the Fock basis using the concepts
of elementary bosonic bricks and composite
fermionic bricks \cite{korcyl1} (see tables \ref{tab. su3_bosonic_bricks} and \ref{tab. su3_fermionic_bricks}).
Elementary bricks are linearly independent single trace operators,
composed uniquely of creation operators, which cannot be reduced by the Cayley-Hamilton theorem.
The set of states obtained by acting with all possible products of powers of bosonic elementary bricks
on the Fock vacuum spans the bosonic sector of the Hilbert space \cite{doktorat_macka}.
%In fact,
%one does not have to explicitly memorize every Fock basis state, however one must introduce an unique labeling of these states.
As far as the fermionic
sectors are concerned, apart of the bosonic bricks one has to use fermionic bricks which need not to be single traces operators \cite{korcyl1}.
\begin{table}
\begin{center}
\begin{tabular}{|c|}
\hline
$n_F=0$\\
\hline
$(a^{\dagger}a^{\dagger})$ \\
$(a^{\dagger}a^{\dagger}a^{\dagger})$ \\
\hline
\end{tabular}
\end{center}
\caption{$SU(3)$ bosonic bricks. \label{tab. su3_bosonic_bricks}}
\end{table}
\begin{table}
\begin{center}
\begin{tabular}{|c|c|c|c|}
\hline
$n_F=1$ & $n_F=2$ & $n_F=3$ & $n_F=4$\\
\hline
 $(f^{\dagger} a^{\dagger})$& $(f^{\dagger}f^{\dagger} a^{\dagger})$ & $(f^{\dagger}f^{\dagger}f^{\dagger})$ & $(f^{\dagger}f^{\dagger}f^{\dagger}f^{\dagger}a^{\dagger})$\\
 $(f^{\dagger}a^{\dagger} a^{\dagger})$ & $(f^{\dagger}f^{\dagger} a^{\dagger}a^{\dagger})$& $(f^{\dagger}f^{\dagger} f^{\dagger}a^{\dagger})$& $(f^{\dagger}a^{\dagger})(f^{\dagger}f^{\dagger}f^{\dagger})$\\
 & $(f^{\dagger} a^{\dagger}a^{\dagger}f^{\dagger}a^{\dagger})$& $(f^{\dagger} f^{\dagger}f^{\dagger}a^{\dagger}a^{\dagger})$& $(f^{\dagger}f^{\dagger}f^{\dagger}f^{\dagger}a^{\dagger}a^{\dagger})$\\
 & $(f^{\dagger}a^{\dagger})(f^{\dagger} a^{\dagger}a^{\dagger})$& $(f^{\dagger}a^{\dagger})(f^{\dagger} f^{\dagger}a^{\dagger})$&$(f^{\dagger}a^{\dagger}a^{\dagger})(f^{\dagger}f^{\dagger}f^{\dagger})$\\
 && $(f^{\dagger}a^{\dagger}f^{\dagger}f^{\dagger}a^{\dagger}a^{\dagger})$ &$(f^{\dagger}a^{\dagger})(a^{\dagger}f^{\dagger}f^{\dagger}f^{\dagger})$\\
 && $(f^{\dagger}a^{\dagger})(f^{\dagger}f^{\dagger}a^{\dagger}a^{\dagger})$ & $(f^{\dagger}f^{\dagger}a^{\dagger})(f^{\dagger}f^{\dagger}a^{\dagger})$\\
 && $(f^{\dagger}a^{\dagger}a^{\dagger})(f^{\dagger}f^{\dagger}a^{\dagger})$ & $(f^{\dagger}a^{\dagger}a^{\dagger})(f^{\dagger}f^{\dagger}f^{\dagger}a^{\dagger})$\\
 && $(f^{\dagger}a^{\dagger}a^{\dagger})(f^{\dagger}f^{\dagger}a^{\dagger}a^{\dagger})$ & $(f^{\dagger}f^{\dagger}a^{\dagger})(f^{\dagger}f^{\dagger}a^{\dagger}a^{\dagger})$\\
 &&& $(f^{\dagger}a^{\dagger})(f^{\dagger}a^{\dagger}a^{\dagger})(f^{\dagger}f^{\dagger}a^{\dagger})$\\
 &&& $(f^{\dagger}f^{\dagger}a^{\dagger})(f^{\dagger}a^{\dagger}f^{\dagger}a^{\dagger}a^{\dagger})$\\
\hline
\end{tabular}
\end{center}
\caption{$SU(3)$ fermionic bricks in the sectors with $n_F=1,\dots,4$ fermions. \label{tab. su3_fermionic_bricks}}
\end{table}

The spectra of the Hamiltonian are obtained by diagonalizing numerically the Hamiltonian matrix evaluated in the
cut basis.
Therefore, one introduces a
cut-off, $N_{cut}$, which limits the maximal number of quanta contained in the
basis states.
Once the cut-off is set the Hamiltonian matrix becomes finite and it is possible to evaluate all its elements.
Subsequently, we can diagonalize numerically such matrix, as well as any other matrices corresponding to other physical observables.
The physically reliable results are those obtained in the infinite cut-off limit, called also the continuum limit \cite{korcyl0}.
The procedure requires to perform several
calculations with increasing cut-off and to extrapolate the results to the continuum limit.
It was shown that the eigenenergies corresponding to localized and nonlocalized states behave differently \cite{maciek2}\cite{maciek1}\cite{korcyl_msc}.
The former ones converge rapidly to their physical limit, whereas the latter ones slowly fall to zero.

% for the localized states are those which no longer depend on the cut-off after making .
% The behavior of nonlocalized states with increasing cut-off is an interesting subject and was studied for the
%$SO(d)$ and $SU(2)$ gauge groups .

Our analytic approach uses the cut Fock basis since it provides a systematic control of the Hilbert space. However,
instead of evaluating
the matrix elements of the Hamiltonian operator, one transforms the eigenvalue problem into a problem of finding
a solution of some recursion relation.

The analytic formulae derived in this article are particularly useful since they
describe both the approximate numerical results and the exact continuum solutions.
This is possible because they are parameterized by the cut-off $N_{cut}$.
In the limit of infinite $N_{cut}$ our formulae give the exact solutions, whereas for any finite cut-off they can be directly compared
with numerical results.

%In the following sections we will apply these theorems to the problem of finding the eigenstates
%of the Hamiltonian eq.\eqref{eq.hamiltonian_general}, first with $N=2$ and then with $N=3$.

\section{Solutions in the bosonic sector}
\label{sec.bosonic}

%Before we attack the Hamiltonian eq.\eqref{eq.hamiltonian2} we will analyze a simpler model with the $SU(2)$ symmetry group.
%The solution of the $SU(2)$ model was first found by Claudson and Halpern \cite{claudson}. It was then rederived
%by Samuel \cite{samuel} and Trzetrzelewski \cite{doktorat_macka}. This model is particularly simple because the
%$SU(2)$ group is isomorphic with the $SO(3)$ group which can be parametrized by spherical coordinates.
%The Hamiltonian eq.\eqref{eq.hamiltonian_general} for $N=2$ reads,
%\begin{equation}
%H = (a^{\dagger}a) + \frac{3}{4} -\frac{1}{2} (a^{\dagger} a^{\dagger}) - \frac{1}{2}(aa).
%\label{eq. hamiltonian su2}
%\end{equation}
%In the following, we will start by deriving the recurrence relation for the coefficients
%of expansion of eigenstates of eq.\eqref{eq. hamiltonian su2} in the Fock basis and solve
%it using the theorems just presented. Exploiting these results we will rederive
%the solutions of Claudson and Halpern. Eventually, we will calculate the Witten index for this model.

In this section we present the derivation of the exact bosonic solutions of the $SU(3)$ SYMQM model.
%We start by transforming the eigenvalue problem into a recurrence relation on the coefficients of
%the decomposition of an eigenstate in the cut Fock basis.
We use theorems presented in appendix \ref{sec.theorems} to solve the recursion relations obtained from the eigenvalue problem.
We classify the solutions into sets and describe their general structure.
For the simplest cases we demonstrate their orthogonality and completeness.

\subsection{Recurrence relation}

A general state $|E\rangle$ from the bosonic sector can be expanded in the basis as
\begin{equation}
|E\rangle = \sum_{2j+3k \le N_{cut}} a_{j,k} \ (a^{\dagger} a^{\dagger})^j  (a^{\dagger} a^{\dagger}a^{\dagger})^k |0\rangle.
\label{eq.expansion}
\end{equation}
It was shown \cite{doktorat_macka} by explicit construction
and by the character method, that the set of states
$\big\{(a^{\dagger} a^{\dagger})^j  (a^{\dagger} a^{\dagger}a^{\dagger})^k |0\rangle \big\}$ containing less then $N_{cut}$
quanta, i.e. ${2j+3k \le N_{cut}}$,
is indeed a complete and linearly independent set of states. It spans the Hilbert space of states with less than $N_{cut}$ bosonic quanta.
However, the Fock basis used in eq.\eqref{eq.expansion} is not orthogonal, which will have some implications for the structure of the solutions.

The requirement that $|E\rangle$ is an eigenstate of the Hamiltonian to the eigenenergy $E$ can be translated to the condition
that $a_{j,k}$ must obey the following
recursion relation (for the details of the derivation see appendix \ref{app. 2}),
\begin{multline}
a_{j-1,k} - \big( 2j + 3k + 4 - 2E \big) a_{j,k} + (j + 1)(j + 3k+4) a_{j+1,k}
%\nonumber \\
\\
+ \frac{3}{8}(k+1)(k+2) a_{j-2,k+2} = 0.
\label{eq.su3.recurrence.rel}
\end{multline}
The coefficients $a_{j,k}$ with $j,k<0$ are set to $0$. Note that the first three terms
involve the same value of the second index of $a_{j,k}$, whereas
the last term mixes the coefficients $a_{j,k}$ with different values of $k$.
However, it is clear that since it only mixes $k$'s with the same
parity, we can consider separately the sets of coefficients $a_{j,k}$ with $k$ even and with $k$ odd. We will use this observation
while classifying all solutions of eq.\eqref{eq.su3.recurrence.rel} in section \ref{subsec. Families of solutions}.

The mixing term is a direct consequence of the fact that the basis states containing the same number
of quanta are not orthogonal. Hence, solving the recursion relation eq.\eqref{eq.su3.recurrence.rel}
corresponds to finding an orthogonal basis spanning the entire Hilbert space of states with
less than $N_{cut}$ quanta. Comparing to the simpler model with $SU(2)$ gauge group, this mixing term is a new feature.
This follows from the fact that the $SU(2)$ basis did not contained two linearly independent states with equal number of quanta.
For models with $SU(N>3)$ gauge groups, there will be many such mixing terms.

\subsection{Generalized Laguerre polynomials}
\label{subsec. generalized Laguerre polynomials}

Before we study in detail the solutions of the recursion eq.\eqref{eq.su3.recurrence.rel},
let us briefly summarize some basic properties of the generalized Laguerre polynomials, which will be needed in the following.

The Laguerre polynomials $\mathcal{L}_n^{\alpha}(x)$ are defined
as the solutions of the differential equation
\begin{equation}
x y'' + (\alpha +1 -x) y' + n y = 0,
\end{equation}
and the orthogonality relation
\begin{equation}
\int_0^{\infty} \mathcal{L}_m^{\alpha}(x) \mathcal{L}_n^{\alpha}(x) x^{\alpha} e^{-x} dx = \delta_{m n}.
\end{equation}
They fulfill the following three-term recursion relation,
\begin{equation}
(n+\alpha)\mathcal{L}^{\alpha}_{n-1}(x) - (2n+\alpha+1-x) \mathcal{L}^{\alpha}_n(x) + (n+1)\mathcal{L}^{\alpha}_{n+1}(x) = 0.
\label{eq.laguerre.recurrence.relation}
\end{equation}
In our problem we will encounter a recursion relation for the \emph{rescaled} Laguerre polynomials $L_n^{\alpha}(x)$, sometimes known in the literature
as the Sonine polynomials, which are defined by,
\begin{equation}
L^{\alpha}_n(x) \equiv \frac{\mathcal{L}^{\alpha}_n(x)}{\Gamma(n+\alpha+1)}.
\end{equation}
Then, the recursion relation eq. \eqref{eq.laguerre.recurrence.relation} becomes
\begin{equation}
L^{\alpha}_{n-1}(x) - (2n+\alpha+1-x)L^{\alpha}_n(x) + (n+1)(n+\alpha+1) L^{\alpha}_{n+1}(x) = 0.
\label{eq.rescaled.recurrence.relation}
\end{equation}

%The equation \eqref{eq.rescaled.recurrence.relation} will be called the rescaled, generalized Laguerre \emph{equation},
%when we write it for a single, specific value of the index $n$. Contrary, when we write it for any $n$ and we want to mean by it the set of
%all rescaled, generalized Laguerre equations, we will call it the rescaled, generalized Laguerre \emph{recursion relation}.

Equation \eqref{eq.rescaled.recurrence.relation} will be called the generalized Laguerre \emph{equation},
when we write it for a single, specific value of the index $n$. Contrary, when we write it for any $n$ and we want to mean by it the set of
all generalized Laguerre equations, we will call it the generalized Laguerre \emph{recursion relation}.

Moreover, we call the \emph{cut} generalized Laguerre recursion relation of order $N_{cut}$, the
set of generalized Laguerre equations with $N_{cut}+1$ variables denoted by $a_0(x), a_1(x), \dots, a_{N_{cut}}(x)$.

To simplify the notation we introduce a vector notation for
the components $a_{j}(x)$ related by the Laguerre three term recursion relation eq. \eqref{eq.rescaled.recurrence.relation}.
Thus we will write
\begin{align}
\textbf{S}^{\alpha}_{j}(x) \cdot \textbf{a}_{j}(x) \equiv a_{j-1}(x) - \big( 2j + \alpha + 1 - x \big) a_{j}(x) +(j+1)(j+\alpha+1) a_{j+1}(x), \nonumber
\end{align}
where by $\textbf{S}^{\alpha}_{j}(x)$ and $\textbf{a}_{j}(x)$ we mean the vectors
\begin{align}
%a_{j-1}(x) - \big( 2j + \alpha + 1 - x \big) a_{j}(x) +(j+1)(j+\alpha+1) a_{j+1}(x) \equiv \nonumber \\
\textbf{S}^{\alpha}_{j}(x) &\equiv \left(
\begin{array}{ccc}
1, & - \big( 2j + \alpha + 1 - x \big), & (j+1)(j+\alpha+1) %\nonumber
\end{array} \right), \nonumber \\
%\end{align}
%\begin{align}
\textbf{a}_{j}(x) &\equiv
\left(
\begin{array}{c}
a_{j-1}(x) \\
a_{j}(x) \\
a_{j+1}(x)
\end{array} \right). \nonumber
\end{align}

\subsection{Families of solutions}
\label{subsec. Families of solutions}

It turns out that the solutions
of eq.\eqref{eq.su3.recurrence.rel} can be naturally classified
into separate sets, which we call \emph{families}. Families can be
labeled by a single integer $\kappa$. For a given family, $\kappa$ is equal to the maximal number of cubic bricks
contained in the basis states used to construct the solutions within this family.
Additionally we adopt a different notation for even and odd solutions. Therefore, we define:
\begin{itemize}
\item a solution belongs to the set $f_{\kappa}$ if $a_{j,k} \equiv 0, k > 2{\kappa}$
and $a_{j,k} \ne 0, k \le 2{\kappa}$. In words, the eigenstate
can be decomposed into basis states containing an even number of cubic bricks and the maximal number of them is $2{\kappa}$,
\item a solution belongs to the set $g_{\kappa}$ if $a_{j,k} \equiv 0, k > 2{\kappa} + 1$
and $a_{j,k} \ne 0$, $k \le 2{\kappa} + 1$. In words, the eigenstate
can be decomposed into basis states containing an odd number of cubic bricks and the maximal number of them is $2{\kappa} + 1$.
\end{itemize}
$f_0$ is the simplest set of solutions, for which only $a_{j,0}$ are nonzero, i.e.
is only composed of bilinear bricks or in other words it involves only the Kronecker delta invariant
tensor.

We introduce now the quantities $d_{\kappa}(N_{cut})$ and $d_{\kappa}'(N_{cut})$ which denote the number
of states in the $f_{\kappa}$ and $g_{\kappa}$ family respectively at a given
cut-off $N_{cut}$.
$d_0 = \lfloor \frac{N_{cut}}{2} \rfloor + 1$ is the number of states composed exclusively of the quadratic elementary brick $(a^{\dagger} a^{\dagger})$,
namely,
\begin{equation}
|0\rangle, \ (a^{\dagger} a^{\dagger})|0\rangle, \ (a^{\dagger} a^{\dagger})^2|0\rangle, \ \dots, \ (a^{\dagger} a^{\dagger})^{d_0-1}|0\rangle. \nonumber
\end{equation}
Hence, the energy eigenstates from the family $f_0$ will correspond to $d_0$ independent linear combinations of such states.

Correspondingly, there will be $d_0' = \lfloor \frac{N_{cut}-3}{2}\rfloor +1$ states composed with exactly one cubic elementary brick, since there
are $d_0'$ such states in the basis, namely,
\begin{equation}
(a^{\dagger}a^{\dagger}a^{\dagger})|0\rangle, \ (a^{\dagger} a^{\dagger})(a^{\dagger}a^{\dagger}a^{\dagger})|0\rangle, \ (a^{\dagger} a^{\dagger})^2(a^{\dagger}a^{\dagger}a^{\dagger})|0\rangle, \ \dots, \ (a^{\dagger} a^{\dagger})^{d'_0-1}(a^{\dagger}a^{\dagger}a^{\dagger})|0\rangle. \nonumber
\end{equation}
In general, $d_{\kappa}$ and $d_{\kappa}'$ are given by
\begin{align}
d_{\kappa}  = \Big\lfloor \frac{1}{2} \big( N_{cut} -6\kappa \big) \Big\rfloor + 1,  \qquad
d_{\kappa}' = \Big\lfloor \frac{1}{2} \big( N_{cut} -6\kappa -3 \big) \Big\rfloor + 1.
\end{align}

Obviously, for a given cut-off $N_{cut}$ there will be in total $\lfloor \frac{N_{cut}}{3} \rfloor$ nontrivial families of solutions.

%\subsection{Generic solutions belonging to families $f_{\kappa}$ and $g_{\kappa}$}
\subsection{Generic solutions}

Let us now consider generic families $f_{\kappa}$ and $g_{\kappa}$.
At a given, finite cut-off $N_{cut}$, the family $f_{\kappa}$ consists of $d_{\kappa}(N_{cut})$ solutions.
Similarly, the family $g_{\kappa}$ consists of $d'_{\kappa}(N_{cut})$ solutions.
For each solution the eigenvector and the eigenenergy must be specified.
We describe below, separately
for the $f_{\kappa}$ and $g_{\kappa}$ family, the expressions for the eigenvectors and the quantization conditions for the eigenenergies.
The details of the derivation are described in the appendix \ref{sec.theorems}.

The form of the normalization factors is motivated by the plane-wave normalization discussed in
section \ref{subsubsec. orthogonality}.

\subsubsection{Even number of cubic bricks}

The solutions in the $f_{\kappa}$ family are given by
\begin{multline}
|E_m,\kappa\rangle_{even} = (2E_m)^{3\kappa} e^{-E_m} \sum_{n=0}^{d_{\kappa}-1} L_n^{6\kappa+3}(2E_m) \Big( |n,2\kappa\rangle +
\\
+ \sum_{t=1}^{\kappa} \Gamma^{even}_{\kappa-1,\kappa-t}|n+3t,2\kappa-2t\rangle \Big),
\end{multline}
where the energies $E_m$ are given by the quantization condition
\begin{equation}
L_{d_{\kappa}}^{6\kappa+3}(2E_m) = 0, \qquad 1 \le m \le d_{\kappa}.
\label{eq. quantization condition}
\end{equation}
$\Gamma(x)$ denotes the Gamma function, $\Gamma(x)=(x-1)!$ for $x$ integer and $\Gamma^{even}_{\kappa-1,t}$ is given by
\begin{equation}
\Gamma^{even}_{\kappa-1,t} = (-24)^{t-\kappa}\frac{(\kappa+t)!}{(\kappa-t)!(2t)!}.
\end{equation}
By introducing a more convenient
notation
\begin{equation}
\overline{|n, 2\kappa \rangle} \equiv |n,2\kappa\rangle + \sum_{t=1}^{\kappa} \Gamma^{even}_{\kappa-1,\kappa-t}|n+3t,2\kappa-2t\rangle,
\label{eq. orthonormalized fock states}
\end{equation}
we can rewrite the eigenstates in a compact form as
\begin{align}
|E_m,\kappa\rangle_{even} = (2E_m)^{3\kappa} e^{-E_m} \sum_{n=0}^{d_{\kappa}-1} & L_n^{6\kappa+3}(2E_m) \overline{|n,2\kappa\rangle}.
\end{align}

\subsubsection{Odd number of cubic bricks}

The quantization condition for this family reads
\begin{equation}
L_{d_{\kappa}'}^{6\kappa+6}(2E_m) = 0, \qquad 1 \le m \le d_{\kappa}',
\end{equation}
and the corresponding eigenstates are given by
\begin{multline}
|E_m,\kappa\rangle_{odd} = (2E_m)^{3\kappa+\frac{3}{2}} e^{-E_m} \sum_{n=0}^{d_{\kappa}'-1} L_n^{6\kappa+6}(2E_m) \Big( |n, 2\kappa+1 \rangle +
\\
+ \sum_{t=1}^{\kappa}  \Gamma^{odd}_{\kappa-1,\kappa-t} |n+3t,2\kappa-2t+1\rangle \Big),
\end{multline}
where
\begin{equation}
\Gamma^{odd}_{\kappa-1,t} = (-24)^{t-\kappa}\frac{(\kappa+t+1)!}{(\kappa-t)!(2t+1)!}.
\end{equation}
Again, using an simplified notation,
\begin{equation}
\overline{|n, 2\kappa+1\rangle} \equiv |n, 2\kappa+1 \rangle + \sum_{t=1}^{\kappa}  \Gamma^{odd}_{\kappa-1,\kappa-t} |n+3t,2\kappa-2t+1\rangle
\end{equation}
we can write $|E_m,\kappa\rangle_{odd}$ as
\begin{align}
|E_m,\kappa\rangle_{odd} = (2E_m)^{3\kappa+\frac{3}{2}} e^{-E_m} \sum_{n=0}^{d_{\kappa}'-1} & L_n^{6\kappa+6}(2E_m) \overline{|n, 2\kappa+1 \rangle}.
\end{align}

In the next section we will use the above results to describe the set of all solutions to the eigenvalue problem at finite cut-off.

\subsection{Complete solutions at finite cut-off}

At a given cut-off $N_{cut}$, the set of \emph{all} solutions, $\big\{|E\rangle \big\}_{N_{cut}}$, of
the eigenequation $H|E\rangle = E|E\rangle$ consists of solutions belonging to different families.
As noted in section \ref{subsec. Families of solutions},
there will be representatives of $\big\lfloor \frac{N_{cut}}{3} \big\rfloor$ distinct families.
Hence, $\big\{|E\rangle \big\}_{N_{cut}}$ contains solutions from families $f_{\kappa}$, where $1 \le \kappa \le \kappa_{max} \equiv \lfloor \frac{1}{6} N_{cut} \rfloor$,
and from families $g_{\kappa'}$, where $1 \le \kappa' \le \kappa'_{max} \equiv \lfloor \frac{1}{6} \big( N_{cut} - 3 \big) \rfloor$.
The dependence on
the cut-off is hidden in the integers $d_{\kappa}(N_{cut})$ and $d'_{\kappa'}(N_{cut})$.
We now provide a more detailed description of $\big\{|E\rangle \big\}_{N_{cut}}$.

Since we have $\overline{|n,0\rangle} = |n,0\rangle$
and $\overline{|n,1\rangle} = |n,1\rangle$ there are two families of solutions for
which there is no mixing, namely $f_0$ and $g_0$. They consist of:
\begin{itemize}
\item $d_0$ solutions with $E_m$ such that $L_{d_0}^3(2E_m) = 0$ given by ($1 \le m \le d_0$)
\begin{equation}
|E_m,0\rangle_{even} = e^{-E_m} \sum_{n=0}^{d_0-1} L_n^3(2E_m) \overline{|n,0\rangle}, %= e^{-E_m} \sum_{n=0}^{d_0-1} L_n^3(2E_m) |n,0\rangle,
\label{eq. rozwiazanie f0}
\end{equation}
\item $d_0'$ solutions with $E_m$ such that $L_{d_0'}^6(2E_m) = 0$ given by ($1 \le m \le d'_0$)
\begin{equation}
|E_m,0\rangle_{odd} = (2E_m)^{\frac{3}{2}} e^{-E_m} \sum_{n=0}^{d_0'-1} L_n^6(2E_m) \overline{|n,1\rangle}, %= (2E)^{\frac{3}{2}} e^{-E_m} \ \sum_{n=0}^{d_0'-1} L_n^6(2E_m) |n,1\rangle.
\label{eq. rozwiazanie g0}
\end{equation}
\end{itemize}

For the remaining $\lfloor \frac{N_{cut}}{3}\rfloor -2 $ families of solutions the mixing appears.
The mixing becomes more and more complex as we consider families with a bigger maximal number of cubic bricks.
We present the eigenstates from the first two families with a simple mixing,
and then we give the expressions for eigenstates from the most complicated families.
Hence, we have
\begin{itemize}
\item $d_1$ solutions with $E_m$ such that $L_{d_1}^9(2E_m) = 0$ given by ($1 \le m \le d_1$)
\begin{align}
|E_m,1\rangle_{even} &= (2E_m)^3 e^{-E_m} \sum_{n=0}^{d_1-1} L_n^9(2E_m) \overline{|n,2\rangle},
%\nonumber \\
%& = \sum_{n=0}^{d_1-1} L_n^9(2E_m) \big( |n,2\rangle - \frac{1}{24} |n+3,0\rangle \big),
\label{eq. rozwiazanie f1}
\end{align}
where $\overline{|n,2\rangle} = |n,2\rangle - \frac{1}{24} |n+3,0\rangle$.
\item $d_1'$ solution with $E_m$ such that $L_{d_1'}^{12}(2E_m) = 0$ are given by ($1 \le m \le d'_1$)
\begin{align}
|E_m,1\rangle_{odd} &= (2E_m)^{\frac{9}{2}} e^{-E_m} \sum_{n=0}^{d_1'-1} L_n^{12}(2E_m) \overline{|n,3\rangle}
%\nonumber \\
%& = e^{-E_m} \sum_{n=0}^{d_1'-1} L_n^{12}(2E_m) \big( |n,3\rangle - \frac{1}{12} |n+3,1\rangle \big),
\label{eq. rozwiazanie g1}
\end{align}
where $\overline{|n,3\rangle} = |n,3\rangle - \frac{1}{12} |n+3,1\rangle$.
\end{itemize}
and
\begin{itemize}
\item $d_{\kappa_{max}}$ solutions with $E_m$ such that $L_{d_{\kappa_{max}}}^{6\kappa_{max}+3}(2E_m) = 0$ given by ($1 \le m \le d_{\kappa_{max}}$)
%, $\kappa_{max} \equiv \lfloor \frac{1}{6} N_{cut} \rfloor$,
\begin{align}
%|E,\kappa\rangle = a_0 \Gamma(6\kappa+4) \sum_{n=0}^{d_{\kappa}-1} & L_n^{6\kappa+3}(2E) \Bigg( |n,2\kappa\rangle + \nonumber \\
%& + \sum_{t=1}^{\kappa} \Gamma^E(\kappa-1,\kappa-t)|n+3t,2\kappa-2t\rangle \Bigg),
|E_m,\kappa_{max}\rangle_{even} = (2E_m)^{3\kappa_{max}} e^{-E_m} \sum_{n=0}^{d_{\kappa_{max}}-1} & L_n^{6\kappa_{max}+3}(2E_m) \overline{|n,2\kappa_{max}\rangle},
\end{align}
\item $d_{\kappa'_{max}}'$ solutions with $E_m$
such that $L_{d_{\kappa'_{max}}'}^{6\kappa'_{max}+6}(2E_m) = 0$ given by ($1 \le m \le d'_{\kappa'_{max}}$)
%, $\kappa'_{max} \equiv \lfloor \frac{1}{6} \big( N_{cut} - 3 \big) \rfloor$,
\begin{align}
%|E,\kappa\rangle = a_0 \Gamma(6\kappa+7) \sum_{n=0}^{d_{\kappa}'-1} & L_n^{6\kappa+6}(2E) \Bigg( |n, 2\kappa+1 \rangle + \nonumber \\
%& + \sum_{t=1}^{\kappa}  \Gamma^O(\kappa-1,\kappa-t) |n+3t,2\kappa-2t+1\rangle \Bigg).
|E_m,\kappa'_{max}\rangle_{odd} = (2E_m)^{3\kappa'_{max}+\frac{3}{2}} e^{-E_m} \sum_{n=0}^{d_{\kappa'_{max}}'-1} & L_n^{6\kappa'_{max}+6}(2E_m) \overline{|n, 2\kappa'_{max}+1 \rangle}
\end{align}
\end{itemize}

The structure of solutions is shown graphically in
Figure \ref{fig. siatka}. Coefficients $a_{j,k}$ are represented with dots on the $j - k$ plane.
Equivalently, each dot corresponds to a basis state constructed with $j$ bilinear and $k$ cubic bosonic bricks.
Hence, a cut-off with
a fixed number of quanta, here $N_{cut}=15$, can be shown as the oblique, straight line.
The states lying below and on this line are included
in the Fock basis, whereas those lying outside are not. The solutions of the recursion relation eq.\eqref{eq.su3.recurrence.rel}
are represented by dotted lines. The lowest line corresponds to the solutions involving only
bilinear bricks i.e. the solutions from the set $f_0$. The dashed line represents solutions from the set $f_1$.
The two horizontal parts of this line denote respectively the amplitudes $a_{j,2}$ and $a_{j,0}$.
The mixing of these amplitudes starts at the number of quanta equal to 6, i.e. both the amplitudes $a_{0,2}$ and $a_{3,0}$ contain 6 quanta.
%
%The straight, red line denotes a cut-off with
%a fixed number of quanta, namely $N_{cut}=15$. The coefficients lying below and on this line are included
%in the Fock basis, whereas coefficients lying outside
%this line are not included. Increasing the cut-off means pushing this line more and more to the right and including
%more and more states into the cut Fock basis.
%As far as the remaining lines are concerned,
%they represent the sets of amplitudes which are related by
%the recursion relation. The lowest line (dotted, blue) corresponds to the solution involving only
%the quadratic bricks i.e. the solution from the $f_0$ family.
%Therefore, only the coefficients $a_{j,0}$ are linked by this line.
%A neighboring horizontal line denotes the solution from the $g_0$ family which involves
%one cubic brick. No mixing for these two solutions occurs.
%On the contrary, the dotted green line links the nonzero coefficients
%forming the eigenstate from the $f_1$ family. It involves coefficients $a_{j,2}$ as well as the coefficients $a_{j+3,0}$. Note that
%the oblique part of this line is parallel to the cut-off line, because only coefficients proportional to basis states
%with the same number of quanta can be mixed, here amplitudes $a_{0,2}$ and $a_{3,0}$ contain both 6 quanta.
\begin{figure}[t]
\begin{center}
\input{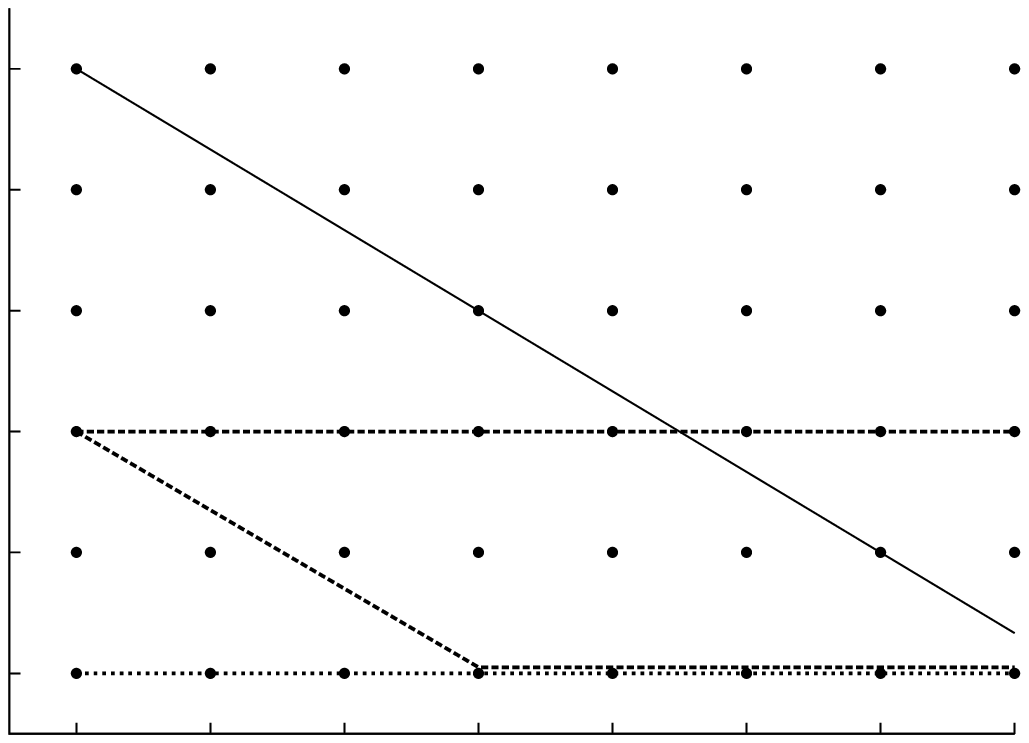}
\caption{The structure of the solutions of the recursion relation eq. \eqref{eq.su3.recurrence.rel}. Each dot represents
a coefficient $a_{j,k}$ with appropriate values of the $j$ and $k$ indices. The solid lines corresponds to a sample
cut-off $N_{cut}=15$, whereas the dotted lines denote solutions as described in the text.
%\caption{The figure represents the $j/k$ plane, and each dot corresponds
%a coefficient $a_{j,k}$ with appropriate values of the $j$ and $k$ indices.
\label{fig. siatka}}
\end{center}
\end{figure}

The complete spectrum, i.e. the set of all values of the $E$ parameter, $\big\{ E \big\}_{N_{cut}}$, for which a nonzero eigenstate exists
can be written %as the union of sets
%\begin{equation}
%\big\{ E \big\} = \bigcup_{m=0}^{\big\lfloor \frac{1}{6} N_{cut} \big\rfloor} \Big\{ L_{d_m}^{6m+3}(2E) = 0\Big\} \ \cup \ \bigcup_{m=0}^{\big\lfloor \frac{1}{6} \big( N_{cut} -3 \big) \big\rfloor} \Big\{ L_{d'_m}^{6m+6}(2E) = 0\Big\}.
%\label{eq. complete eigenenergies}
%\end{equation}
%This can be rewritten
in a compact form by introducing a polynomial $\Theta^{n_F=0}_{N_{cut}}(E)$. The roots of $\Theta^{n_F=0}_{N_{cut}}(E)$ correspond
to the eigenenergies of the cut Hamiltonian at a given cut-off $N_{cut}$,
\begin{equation}
\big\{ E \big\}_{N_{cut}} = \big\{ E: \ \Theta^{n_F=0}_{N_{cut}}(E)=0 \big\}.
\end{equation}
%From eq.\eqref{eq. complete eigenenergies}
The polynomial $\Theta_{N_{cut}}^{n_F=0}(E)$ must be equal to the product of all quantization conditions relevant for a given cut-off, hence
\begin{equation}
\Theta^{n_F=0}_{N_{cut}}(E) = \prod_{k=0}^{\lfloor \frac{1}{3} ( N_{cut} ) \rfloor} L_{\lfloor \frac{1}{2} ( N_{cut}-3k ) \rfloor + 1}^{3k+3}(2E).
\label{eq. complete eigenenergies}
\end{equation}
Thus, the expression eq.\eqref{eq. complete eigenenergies} provides a closed formula for the spectrum of Hamiltonian in the bosonic sector for
any finite cut-off.

\subsection{Continuum solutions}
\label{subsec. Continuum solutions}

Although, for any finite $N_{cut}$ the analytic spectra given by eq.\eqref{eq. complete eigenenergies}
match exactly the results obtained by a numerical diagonalization of the Hamiltonian matrix, in order
to retrieve a physical information one has to perform the continuum limit. Numerically, this conforms
to repeat several computations with increasing $N_{cut}$ and extrapolate the results. On the other hand,
having the analytic solutions we can perform the continuum limit exactly.

Remarkably, the solutions obtained in the preceding section
for a finite cut-off are closely related to the continuum solutions.
%The continuum solutions can be unambiguously obtained from its
%approximate form at finite cut-off.
%\begin{theorem}
%\label{th.czwarte}
%The mixing parameters $\Gamma^{even}(n-1,t)$ and $\Gamma^{odd}(n-1,t)$ do not depend on $\kappa$, where $\kappa$ is the total number
%of recursion relations considered.%, provided $n < \kappa$.
%\end{theorem}
%\begin{proof}
This can be seen as follows.
Let us consider a set of $n$ recursion relations with mixing such as those given by the recurrence eq.\eqref{eq.su3.recurrence.rel}.
Let us assume that we are interested in the solutions from the family $f_{\kappa}$ which have the quantization condition
of the form
$L_{d_{\kappa}(N_{cut})}^{6{\kappa}+3}(2E) = 0$ with ${\kappa} < n$. They are obtained by setting identically to zero
all $a_{j,k} = 0$ with $k>2{\kappa}$. Their mixing coefficients are given by $\Gamma^{even}_{{\kappa}-1,t}$.
Increasing the cut-off is equivalent to the inclusion of new recursion relations. However, as far as the family
of solutions $f_{\kappa}$ is concerned, the only differences occur in the set of possible eigenenergies,
since they will be now given
by the condition $L_{d_{\kappa}(N'_{cut})}^{6{\kappa}+3}(2E) = 0$, with $N'_{cut}$ being the new cut-off. The mixing coefficients
are given exclusively in terms of parameters describing the recursion relations which does not involve the value of the
cut-off. Therefore, they will be unaffected by its change. A similar reasoning can be applied for the solutions from the families $g_{\kappa}$.
%\end{proof}

From the argument above, it follows that in order to obtain the formulae for the continuum solutions, one can simply extend
the sum over the basis states to infinity. Hence, for generic eigenstates from families $f_{\kappa}$ and $g_{\kappa}$ we have
\begin{align}
%|E,\kappa\rangle = a_0 \Gamma(6\kappa+4) \sum_{n=0}^{\infty} & L_n^{6\kappa+3}(2E) \Bigg( |n,2\kappa\rangle + \nonumber \\
%& + \sum_{t=1}^{\kappa} \Gamma^E(\kappa-1,\kappa-t)|n+3t,2\kappa-2t\rangle \Bigg)
&|E,\kappa\rangle_{even} = (2E)^{3\kappa} e^{-E} \sum_{n=0}^{\infty}  L_n^{6\kappa+3}(2E) \overline{|n,2\kappa\rangle}
\label{eq. continuum solutions even} \\
%|E,\kappa\rangle = a_0 \Gamma(6\kappa+7) \sum_{n=0}^{\infty} & L_n^{6\kappa+6}(2E) \Bigg( |n, 2\kappa+1 \rangle + \nonumber \\
%& + \sum_{t=1}^{\kappa}  \Gamma^O(\kappa-1,\kappa-t) |n+3t,2\kappa-2t+1\rangle \Bigg)
&|E,\kappa\rangle_{odd} = (2E)^{3\kappa+\frac{3}{2}} e^{-E} \sum_{n=0}^{\infty}  L_n^{6\kappa+6}(2E) \overline{|n, 2\kappa+1 \rangle}
\label{eq. continuum solutions odd}
\end{align}

Some of the properties of these solutions are summarized in the next section.

\subsection{Properties of $|E, \kappa\rangle_{even}$ solutions}

In this section we verify the orthogonality
%, normalization
and completeness
of the solutions derived in the preceding section.
Since we want to demonstrate the mechanisms that
are responsible for these properties we will only deal with even solutions. Therefore we neglect their subscript ${}_{even}$.
The discussion for the odd solutions follows the same lines.
%belonging to two simplest families $f_0$ and $f_1$. The calculations for other families
%do not involve new concepts however are much more tedious computationally.

\subsubsection{Orthogonality}
\label{subsubsec. orthogonality}

The scalar product of two eigenstates of the Hamiltonian reads
\begin{multline}
\langle E,\alpha| E',\beta \rangle = e^{-E-E'} (2E)^{3\alpha} (2E')^{3\beta} \ \times
%\nonumber \\
%&=
\\
\times \ \sum_{n=0}^{d_{\alpha}-1}\sum_{m=0}^{d_{\beta}-1}
\frac{\mathcal{L}^{6\alpha+3}_{n}(2E) }{\Gamma(n+6\alpha+4)}\frac{\mathcal{L}^{6\beta+3}_{m}(2E')}{\Gamma(m+6\beta+4)}
\overline{\langle n,2\alpha} | \overline{m,2\beta\rangle}.
%\nonumber \\
%&\Bigg( \Big( \langle n,2\alpha| + \sum_{p=1}^{\alpha} \Gamma(\alpha-1,p) \langle n+3p,2\alpha-2p| \Big)
%\nonumber \\
       %&
%\Big( |m,2\beta\rangle + \sum_{t=1}^{\beta} \Gamma(\beta-1,t) | m+3t,2\beta-2t\rangle \Big) \Bigg).
\end{multline}
The general expression for the scalar products of orthogonalized basis states
$\overline{\langle n,2\alpha} | \overline{m,2\beta\rangle}$,
is not known. However, one can calculate them for few simplest cases (namely, $\alpha=0,1,2$). We find that
\begin{align}
\overline{\langle n, 2\alpha} | \overline{m, 2\beta\rangle} = \delta_{\alpha \beta} \delta_{n m} \mathcal{N}^{\alpha}(n),
\label{eq. normalization constant}
\end{align}
and
\begin{align}
\mathcal{N}^{0}(n) &= \frac{\Gamma(n+1)\Gamma(n+4)}{6}, \nonumber \\
\mathcal{N}^{1}(n) &= \frac{\Gamma(n+1)\Gamma(n+10)}{3456}, \nonumber \\
\mathcal{N}^{2}(n) &= \frac{\Gamma(n+1)\Gamma(n+16)}{1990656}, \nonumber
\end{align}
which can be summarized in a compact form as
\begin{equation}
\mathcal{N}^{\alpha}(n) = \frac{1}{6}\frac{1}{24^{2\alpha}} \Gamma(n+1)\Gamma(n+4+ 6\alpha). \nonumber
\end{equation}
We can immediately conclude that solutions belonging to different families are orthogonal, even if they have the same energy.

%\begin{align}
%\langle n,0| m, 0 \rangle &= \delta_{m n} \frac{1}{6} \Gamma(n+1)\Gamma(n+4), \nonumber \\
%\langle n,2| m, 2 \rangle &= \delta_{m n} \frac{1}{1728} \Gamma(n+1)\Gamma(n+7) (n^3+15n^2+101n+255), \nonumber \\
%\langle n,4| m, 4 \rangle &= \delta_{m n}  \frac{1}{995328} \Gamma(n+1)\Gamma(n+10) (7n^6 + 255n^5 + 4087n^4 + 38157n^3 + 236710n^2 \nonumber \\
%& + 960192n + 1838880), \nonumber \\
%\langle n,2| m+3, 0 \rangle &= \delta_{m n} \frac{1}{144} \Gamma(n+4)\Gamma(n+7), \nonumber \\
%\langle n,4| m+6, 0 \rangle &= \delta_{m n} \frac{1}{1728} \Gamma(n+7)\Gamma(n+10), \nonumber \\
%\langle n,4| m+3, 2 \rangle &= \delta_{m n} \frac{1}{82944} \Gamma(n+4)\Gamma(n+10) (5n^3+129n^2+1234n+4200),
%\label{eq. iloczyny skalarne Focka}
%\end{align}
%One can conjecture a general form of the normalization constant $\mathcal{N}^{\alpha}(n)$ to be
%\begin{equation}
%\mathcal{N}^{\alpha}(n) = \frac{1}{6}\frac{1}{24^{2\alpha}} \Gamma(n+1)\Gamma(n+4+ 6\alpha). \nonumber
%\end{equation}

As far as the orthogonality of solutions within each family is concerned, we have
\begin{align}
\langle E,\alpha| E',\beta \rangle &=
%a_0^2 \sum_{n,m=0}^{d_{\alpha}-1}
%\mathcal{L}^{6\alpha+3}_{n}(2E) \mathcal{L}^{6\alpha+3}_{m}(2E')
%\frac{\Gamma(6\alpha+4)}{\Gamma(n+6\alpha+4)}\frac{\Gamma(6\alpha+4)}{\Gamma(m+6\alpha+4)}
%\overline{\langle n,2\alpha}|\overline{m,2\alpha \rangle} \nonumber \\
%& =
\delta_{\alpha \beta} e^{-E-E'} (4EE')^{3\alpha} \frac{1}{6}\frac{1}{24^{2\alpha}} \sum_{m=0}^{d_{\alpha}-1} \frac{\mathcal{L}_m^{6\alpha+3}(2E)\mathcal{L}_m^{6\alpha+3}(2E')}{\Gamma(m+6\alpha+4)}m!
\label{eq. scalar product 1}
\end{align}
%According to the following formula
%\begin{align}
%\sum_{m=0}^n \frac{m!}{\Gamma(m+\alpha+1)}& \mathcal{L}_m^{\alpha}(x) \mathcal{L}_m^{\alpha}(y) =
%%\nonumber \\
%%&
%\frac{(n+1)!}{\Gamma(n+\alpha+1)} \frac{1}{x-y} \Bigg[ \mathcal{L}_n^{\alpha}(x) \mathcal{L}_{n+1}^{\alpha}(y)
%- \mathcal{L}_{n+1}^{\alpha}(x) \mathcal{L}_n^{\alpha}(y)\Bigg],
%\label{eq. wzor skonczony}
%\end{align}
The finite sum in eq.\eqref{eq. scalar product 1} can be calculated \cite{abramowitz} and yields %(for $E \ne E'$)
\begin{multline}
\langle E,\alpha | E',\beta \rangle = \delta_{\alpha \beta} e^{-E-E'} (4EE')^{3\alpha} \frac{1}{6}\frac{1}{24^{2\alpha}} \frac{(d_{\alpha})!}{\Gamma(d_{\alpha}+6\alpha+4)} \frac{1}{2( E - E' )} \ \times \\
\times \ \Big(
\mathcal{L}_{d_{\alpha}-1}^{6\alpha+3}(2E) \mathcal{L}_{d_{\alpha}}^{6\alpha+3}(2E') -
\mathcal{L}_{d_{\alpha}}^{6\alpha+3}(2E) \mathcal{L}_{d_{\alpha}-1}^{6\alpha+3}(2E')
\Big).
\end{multline}
Hence, for $E$ and $E'$ fulfilling the quantization conditions,
$\mathcal{L}_{d_{\alpha}}^{6\alpha+3}(2E')=0$ and $\mathcal{L}_{d_{\alpha}}^{6\alpha+3}(2E)=0$, the scalar products reads
\begin{equation}
\langle E,\alpha | E',\beta \rangle = \delta_{\alpha \beta} \delta_{E E'} \frac{1}{12}\frac{1}{24^{2\alpha}} \mathcal{L}_{d_{\alpha}-1}^{6\alpha+3}(2E) \mathcal{L}_{d_{\alpha}-1}^{6\alpha+4}(2E).
\end{equation}
The fact that the eigenstates are orthogonal among each other
is a simple confirmation that we managed to construct a set of orthogonal states in the Hilbert space with even number of quanta.\\

The question of orthogonality of the continuum solutions is more subtle.
The continuum limit of even solutions is given by eq.\eqref{eq. continuum solutions even}.
%\begin{align}
%|E,\alpha\rangle &= a_0 \Gamma(6\alpha+4) \sum_{n=0}^{\infty} L_n^{6\alpha+3}(2E) \overline{|n,2\alpha \rangle}, \nonumber
%\end{align}
%and $E$ can be any positive, real number.
Following the same steps as for calculations with finite cut-off
and using the formula \cite{abramowitz} for an infinite series of product of two Laguerre polynomials we obtain
%\begin{align}
%\sum_{m=0}^{\infty} \frac{m!}{\Gamma(m+\alpha+1)}& \mathcal{L}_m^{\alpha}(x) \mathcal{L}_m^{\alpha}(y)z^m =
%%\nonumber \\
%%&
%(1-z)^{-1} \exp\big(-z \frac{x+y}{1-z} \big) (x y z)^{-\frac{\alpha}{2}} I_{\alpha}\Big( \frac{2\sqrt{x y z}}{1-z} \Big),
%\label{eq. wzor nieskonczony}
%\end{align}
%we have
\begin{multline}
\langle E,\alpha | E',\beta \rangle = \delta_{\alpha \beta} \frac{1}{6}\frac{1}{24^{2\alpha}} e^{-E-E'} (4EE')^{3\alpha} \ \times \\
%\nonumber \\
%&=
\times \ \lim_{z \rightarrow 1^-} \Big[ (1-z)^{-1} \exp\big(-z \frac{2(E+E')}{1-z} \big) (4 E E'z)^{-\frac{6\alpha+3}{2}} I_{6\alpha+3}\Big( \frac{4\sqrt{E E' z}}{1-z} \Big) \Big]
\end{multline}
Introducing $\epsilon=\frac{1}{4}(1-z)$ and exploiting the asymptotic behavior of the Bessel function for large argument, namely
\begin{equation}
I_{\alpha}(x) = \frac{e^x}{\sqrt{2\pi x}}, \qquad \textrm{for large }x,
\end{equation}
we get
\begin{align}
%\langle E,\alpha| E',\alpha \rangle &= \lim_{t \rightarrow \infty} \ t^2 e^{-t^2 2(E+E')} (4 E E')^{-\frac{6\alpha+3}{2}} e^{ t^2 4\sqrt{E E'}} \frac{1}{\sqrt{2 \pi t^2 4\sqrt{E E'}}} \nonumber \\
%&= \frac{1}{\sqrt{4\pi}} \ \lim_{t \rightarrow \infty} \ t e^{-t^2 \big(\sqrt{2E}-\sqrt{2E'}\big)^2 } (4 E E')^{-\frac{12\alpha+7}{4}}.
\langle E,\alpha| E',\beta \rangle &=  \delta_{\alpha \beta} \frac{1}{2} (4 E E')^{-\frac{3}{2}} \frac{1}{6}\frac{1}{24^{2\alpha}} \lim_{\epsilon \rightarrow 0^+} \ \frac{1}{2\sqrt{\pi \epsilon}} e^{ - \frac{1}{4\epsilon}(\sqrt{2E} - \sqrt{2E'})^2 }
\label{eq. scalar product derivation}
\end{align}
Eq.\eqref{eq. scalar product derivation} reduces thanks to the well-known Dirac-delta representation
\begin{equation}
\lim_{\epsilon \rightarrow 0^+} \ \frac{1}{2\sqrt{\pi \epsilon}} e^{ - \frac{x^2}{4\epsilon} } = \delta(x)
\end{equation}
to the desired result
%\begin{align}
%\langle E,\alpha | E',\alpha \rangle = \frac{1}{\sqrt{4\pi}} (2 E)^{-\frac{12\alpha+7}{2}} \ \lim_{t \rightarrow \infty} \ t, \ \  &\textrm{ for } E = E', \nonumber \\
%\langle E,\alpha | E',\alpha \rangle = 0, \qquad \qquad \qquad \qquad \qquad &\textrm{ for } E \ne E'.
%\end{align}
\begin{align}
\langle E,\alpha | E',\beta \rangle = \delta_{\alpha \beta}  \frac{1}{12}\frac{1}{24^{2\alpha}} (2 E)^{-3} \delta(\sqrt{2E} - \sqrt{2E'}).
\end{align}
The norm of these eigenstates is divergent, which is in agreement with our expectation that our solutions in the continuum limit should behave like
plane-waves.

\subsubsection{Completeness}

Next we show that the set of states derived above is a correct basis of the Hilbert space.
Such set of solutions must be complete, i.e. any other solution can be written as a linear combination of $|E, \kappa \rangle$.
This can be done by showing that any Fock state can be expressed in terms of the energy solutions. Hence,
the matrix representation of the transformation of the set of Fock states into the set of solutions $|E, \kappa \rangle$, denoted by $T$,
must be nonsingular.

$T$ can be written as a product of two matrices, $T=T_2 T_1$, where $T_1$ represents the matrix representation of the transformation
$\big\{ | j, k \rangle \big\} \rightarrow \big\{ \overline{| j, \kappa \rangle} \big\}$ , and $T_2$ represents the matrix representation of the
transformation $\big\{ \overline{| j, \kappa \rangle} \big\} \rightarrow \big\{ | E_m, \kappa \rangle \big\}$.

It can be seen from the definition of the orthogonalized Fock states eq.\eqref{eq. orthonormalized fock states} that the matrix $T_1$ is a triangular
matrix with all diagonal elements equal to $1$. Hence, its determinant is simply $1$.

The matrix $T_2$ is block diagonal with each block corresponding to a given family of energy eigenstates $| E_m, \kappa \rangle$.
Hence, the determinant of $T_2$ will be given by a product of determinants of the transformation matrices written for each family independently.
To calculate the latter at a finite cut-off, let us consider $d_{\kappa}$ solutions from the $f_{\kappa}$ family. By definition, they can
be obtained from the orthogonalized basis states through the following matrix
\begin{equation}
\left( \begin{array}{c}
|E_1,\alpha\rangle \\
|E_2,\alpha\rangle \\
\vdots \\
|E_{d_{\kappa}},\alpha\rangle \end{array} \right)
=
\left( \begin{array}{ccc}
%L^{6\alpha+3}_0(2E_1) & L^{6\alpha+3}_1(2E_1) & \dots & L^{6\alpha+3}_{m-1}(2E_1) \\
%L^{6\alpha+3}_0(2E_2) & L^{6\alpha+3}_1(2E_2) & \dots & L^{6\alpha+3}_{m-1}(2E_2) \\
%\vdots &  &  &  \vdots \\
%L^{6\alpha+3}_0(2E_m) & L^{6\alpha+3}_1(2E_m) & \dots & L^{6\alpha+3}_{m-1}(2E_m)
L^{6\kappa+3}_0(2E_1) &  \dots & L^{6\kappa+3}_{d_{\kappa}-1}(2E_1) \\
L^{6\kappa+3}_0(2E_2) &  \dots & L^{6\kappa+3}_{d_{\kappa}-1}(2E_2) \\
\vdots &  &  \vdots \\
L^{6\kappa+3}_0(2E_{d_{\kappa}}) &  \dots & L^{6\kappa+3}_{d_{\kappa}-1}(2E_{d_{\kappa}})
\end{array} \right)
\left( \begin{array}{c}
\overline{|0,2\kappa \rangle} \\
\overline{|1,2\kappa \rangle} \\
\vdots \\
\overline{|d_{\kappa}-1,2\kappa \rangle} \\
\end{array} \right) \nonumber
\end{equation}
One can show that the corresponding determinant is the Vandermonde determinant depending on energies $E_m$. To see this let us
write the Laguerre polynomials $\mathcal{L}^{6\kappa+3}_n(E)$ as
\begin{equation}
L^{6\kappa+3}_n(2E) = \frac{1}{n! \Gamma(n+6\kappa+4)} (2E)^n + \sum_{k=0}^{n-1} c_{n,k} L^{6\kappa+3}_k(2E) \nonumber
%L^{6\kappa+3}_n(2E) = \frac{1}{n! \Gamma(n+6\kappa+4)} (2E)^n + \sum_{k=0}^{n-1} c_{n,k} (2E)^k \nonumber
\end{equation}
where $c_{n,k}$ are some constants. Then, adding to each column an appropriate linear combination
of preceding columns, starting by the last one, and ending on the second one, we get the determinant
\begin{multline}
%&\left| \begin{array}{ccccc}
%\frac{1}{\Gamma(6\kappa+4)} & \frac{1}{\Gamma(6\kappa+5)} 2E_1 & \frac{1}{2 \Gamma(6\kappa+6)}(2E_1)^2 & \dots & \frac{1}{(d_{\kappa}-1)! \Gamma(6\kappa+3+m)}(2E_1)^{d_{\kappa}-1} \\
%\frac{1}{\Gamma(6\kappa+4)} & \frac{1}{\Gamma(6\kappa+5)} 2E_2 & \frac{1}{2 \Gamma(6\kappa+6)}(2E_2)^2 & \dots & \frac{1}{(d_{\kappa}-1)! \Gamma(6\kappa+3+m)}(2E_2)^{d_{\kappa}-1} \\
% & \vdots &  & \vdots &  \\
%\frac{1}{\Gamma(6\kappa+4)} & \frac{1}{\Gamma(6\kappa+5)} 2E_{d_{\kappa}} & \frac{1}{2 \Gamma(6\kappa+6)}(2E_{d_{\kappa}})^2 & \dots & \frac{1}{(d_{\kappa}-1)! \Gamma(6\kappa+3+m)}(2E_{d_{\kappa}})^{d_{\kappa}-1}
%\end{array} \right| = \nonumber \\
\left| \begin{array}{cccc}
\frac{1}{\Gamma(6\kappa+4)} & \frac{1}{\Gamma(6\kappa+5)} 2E_1  & \dots & \frac{1}{(d_{\kappa}-1)! \Gamma(6\kappa+3+m)}(2E_1)^{d_{\kappa}-1} \\
\frac{1}{\Gamma(6\kappa+4)} & \frac{1}{\Gamma(6\kappa+5)} 2E_2  & \dots & \frac{1}{(d_{\kappa}-1)! \Gamma(6\kappa+3+m)}(2E_2)^{d_{\kappa}-1} \\
 & \vdots  & \vdots &  \\
\frac{1}{\Gamma(6\kappa+4)} & \frac{1}{\Gamma(6\kappa+5)} 2E_{d_{\kappa}}  & \dots & \frac{1}{(d_{\kappa}-1)! \Gamma(6\kappa+3+m)}(2E_{d_{\kappa}})^{d_{\kappa}-1}
\end{array} \right| = \nonumber \\
= \Big( \prod_{k=0}^{d_{\kappa}-1} \frac{1}{k! \ \Gamma(6\kappa+4+k)} \Big) \prod_{1 \le i < j \le d_{\kappa}} (2E_j - 2E_i)
\end{multline}
Recall, that the energies $E_m$ are given by the quantization condition $\mathcal{L}_{d_{\kappa}}^{6\kappa+3}(2E) = 0$.
The roots of the Laguerre polynomial are simple, i.e. the set of roots does not contain two equal numbers. Hence, the above
Vandermonde determinant is nonzero.

In order to extend this result to the continuum limit
it is necessary to prove that the Laguerre polynomials in the limit of infinite order have only simple roots. This can be done using
the following formula \cite{abramowitz},
\begin{equation}
\lim_{n \rightarrow \infty} n^{-6\kappa-3} \mathcal{L}_n^{6\kappa+3} \big(\frac{z}{n}\big) = z^{-\frac{6\kappa+3}{2}} J_{6\kappa+3}(2 \sqrt{z}).
\end{equation}
The Bessel function of the first kind, $J_{6\kappa+3}(2 \sqrt{z})$, has countably many simple zeros.
It follows that the zeros of the Laguerre polynomials in the limit of infinite order are also simple.
Therefore, the infinite Vandermonde determinant will be also nonzero.

The determinant of the $T$ matrix is the product of determinants of $T_1$ and $T_2$ matrices. Hence, we showed
that it is nonzero for both finite and infinite cut-off.
This proves that the transformation matrix $T$ is nonsingular and therefore the set of solutions $\big\{|E_m,\kappa\rangle\big\}$
is a correct, orthogonal basis of the Hilbert space for the system with the $SU(3)$ gauge symmetry.

\section{General recurrence relation in fermionic sectors}
\label{sec. General recurrence relation in the fermionic sectors of the $SU(3)$ model}

In this section we derive the recursion relation for the components of energy eigenstates in the Fock basis in the fermionic sectors.
We find a general recurrence in a sector with $n_F$ fermions, and then specialize to the case of $n_F=1$ in the following section.

In order to proceed we must introduce some notation for the fermionic bricks
such as those presented in table \ref{tab. su3_fermionic_bricks}. Let us denote
a generic fermionic brick in the sector with $n_F$ fermionic quanta by $C^{\dagger}(n_B^{\alpha},n_F,\alpha)$.
The index $\alpha$ labels all fermionic bricks in this sector and takes values in $1 \le \alpha \le d^{n_F}$, where $d^{n_F}$
is the total number of fermionic
bricks the $n_F$ sector. $n_B^{\alpha}$ designs the number of bosonic creation operators contained in the $\alpha$-th fermionic brick.

The basis in the sector with $n_F$ fermions is composed of the following states
\begin{equation}
|j, k, \alpha \rangle \equiv C^{\dagger}(n_B^{\alpha}, n_F, \alpha) \ (a^{\dagger} a^{\dagger})^j (a^{\dagger} a^{\dagger} a^{\dagger})^k |0 \rangle,
\end{equation}
Hence, a general state can be decomposed as
\begin{equation}
|E\rangle = \sum_{\alpha=1}^{d^{n_F}} \ \sum_{2j+3k \le N_{cut}-n_B^{\alpha}} a^{\alpha}_{j,k} \ | j, k, \alpha \rangle
%F_i(n_F) \ (a^{\dagger} a^{\dagger})^j (a^{\dagger} a^{\dagger} a^{\dagger})^k |0 \rangle. \nonumber
\end{equation}
Let us rewrite the eigenequation for $H$ in terms of the coefficients $a^{\alpha}_{j,k}$, namely
\begin{multline}
%\big( H - E \mathbb{I} \big)|E\rangle &= \sum_{i=1}^{d^{n_F}} \sum_{2j+3k \le N_{cut}}
%\ a^i_{j,k} \ \big[ H, F_i(n_F) \big] \ |j,k\rangle \nonumber \\%(a^{\dagger} a^{\dagger})^j (a^{\dagger} a^{\dagger} a^{\dagger})^k |0 \rangle \nonumber \\
%&+
%\sum_{i=1}^{d^{n_F}} \sum_{2j+3k \le N_{cut}} a^i_{j,k} \ F_i(n_F)\  H \ |j,k\rangle %(a^{\dagger} a^{\dagger})^j (a^{\dagger} a^{\dagger} a^{\dagger})^k |0 \rangle
%= 0. \nonumber
\big( H - E \big)|E\rangle = \\  \sum_{\alpha=1}^{d^{n_F}} \sum_{2j+3k \le N_{cut} - n_B^{\alpha}}
\ a^{\alpha}_{j,k} \ \Big( \big[ H, C^{\dagger}(n_B^{\alpha}, n_F, \alpha) \big] %\ |j,k\rangle \nonumber \\%(a^{\dagger} a^{\dagger})^j (a^{\dagger} a^{\dagger} a^{\dagger})^k |0 \rangle \nonumber \\
%\sum_{i=1}^{d^{n_F}} \sum_{2j+3k \le N_{cut}} a^i_{j,k} \
+ C^{\dagger}(n_B^{\alpha}, n_F, \alpha) \  H \Big) |j,k\rangle %(a^{\dagger} a^{\dagger})^j (a^{\dagger} a^{\dagger} a^{\dagger})^k |0 \rangle
= 0. \nonumber
\end{multline}
Evaluation of the commutator yields
\begin{equation}
\big[ H, C^{\dagger}(n_B^{\alpha}, n_F, \alpha) \big] %= \big[ (a^{\dagger}a) - \frac{1}{2} (aa) , C^{\dagger}(n_B^{\alpha}, n_F, \alpha) \big]
= \frac{1}{2}n_B^{\alpha}  C^{\dagger}(n_B^{\alpha}, n_F, \alpha) - \frac{1}{2} \big[ (aa) , C^{\dagger}(n_B^{\alpha}, n_F, \alpha) \big]
\end{equation}
We must now evaluate the commutator $\big[ (aa) , C^{\dagger}(n_B^{\alpha}, n_F, \alpha) \big]$.
For a particular $C^{\dagger}(n_B^{\alpha}, n_F, \alpha)$ it will be given by a sum of $n_B^{\alpha}$ terms,
each equal to
$C^{\dagger}(n_B^{\alpha}, n_F, \alpha)$ with one of the bosonic creation operators replaced
by a bosonic annihilation operator. Let us denote these terms by $G^t_{\alpha}$,
where the index $t$ runs from $1$ to $n_B^{\alpha}$,
\begin{equation}
\big[ (aa) , C^{\dagger}(n_B^{\alpha}, n_F, \alpha) \big] = \sum_{t=1}^{n_B^{\alpha}} G^t_{\alpha}.
\end{equation}
In order to get rid of the annihilation operators contained in $G^t_{\alpha}$, we would like to push them over the creation operators
$(a^{\dagger} a^{\dagger})^j (a^{\dagger} a^{\dagger} a^{\dagger})^k$,
so that they hit the Fock vacuum. We have
\begin{equation}
\forall_t \ \big[ G^t_{\alpha}, (a^{\dagger} a^{\dagger})^j (a^{\dagger} a^{\dagger} a^{\dagger})^k \big] =
(a^{\dagger} a^{\dagger})^j \big[ G^t_{\alpha}, (a^{\dagger} a^{\dagger} a^{\dagger})^k \big] +
\big[ G^t_{\alpha}, (a^{\dagger} a^{\dagger})^j \big] (a^{\dagger} a^{\dagger} a^{\dagger})^k
\end{equation}
Since, $G^t_{\alpha}$ contains exactly one annihilation operator, one can simplify the above commutators, as
\begin{align}
\forall_t \ \big[ G^t_{\alpha}, (a^{\dagger} a^{\dagger})^j \big] &= j \ (a^{\dagger} a^{\dagger})^{j-1} \big[ G^t_{\alpha}, (a^{\dagger} a^{\dagger})\big],  \\
\forall_t \ \big[ G^t_{\alpha}, (a^{\dagger} a^{\dagger} a^{\dagger})^k \big] &=
k \ (a^{\dagger} a^{\dagger} a^{\dagger})^{k-1} \big[ G^t_{\alpha}, (a^{\dagger} a^{\dagger} a^{\dagger})\big].
\end{align}
Therefore,
\begin{align}
\forall_t \ \big[ G^t_{\alpha}, (a^{\dagger} a^{\dagger})^j (a^{\dagger} a^{\dagger} a^{\dagger})^k \big] & =
k \ \big[ G^t_{\alpha}, (a^{\dagger} a^{\dagger} a^{\dagger})\big] \ (a^{\dagger} a^{\dagger})^j  (a^{\dagger} a^{\dagger} a^{\dagger})^{k-1}  \nonumber \\
& + j \ \big[ G^t_{\alpha}, (a^{\dagger} a^{\dagger})\big] \ (a^{\dagger} a^{\dagger})^{j-1} (a^{\dagger} a^{\dagger} a^{\dagger})^k
\end{align}
Eventually, it is easy to evaluate $\big[ G^t_{\alpha}, (a^{\dagger} a^{\dagger})\big]$ since this simply yields
back the $C^{\dagger}(n_B^{\alpha}, n_F, \alpha)$ composite fermionic brick. Thus, we obtain
%\begin{align}
%\sum_{\alpha=1}^{d^{n_F}}& \sum_{2j+3k \le N_{cut}-n_B^{\alpha}} \Big\{ C^{\dagger}(n_B^{\alpha}, n_F, \alpha) \Big( \textbf{S}^{3k+3+n_B^{\alpha}}
%\cdot \textbf{a}^{\alpha}_{j,k} \nonumber \\ &+ \frac{3}{8}(k+1)(k+2) a^{\alpha}_{j-2,k+2} \Big) |j,k\rangle \nonumber \\
%%(a^{\dagger} a^{\dagger})^j (a^{\dagger} a^{\dagger} a^{\dagger})^k |0 \rangle \nonumber \\
%%& + \sum_{\alpha=1}^{d^{n_F}} \sum_{t=1}^{n_B^{\alpha}} \sum_{2j+3k \le N_{cut}-n_B^{\alpha}} a^{\alpha}_{j,k} \big( k
%& +  \sum_{t=1}^{n_B^{\alpha}} a^{\alpha}_{j,k} \big( k
%\ \big[ G^t_{\alpha}, (a^{\dagger} a^{\dagger} a^{\dagger})\big] \ |j,k-1\rangle + (a^{\dagger} a^{\dagger})^j (a^{\dagger} a^{\dagger} a^{\dagger})^k G^t_{\alpha} |0 \rangle \big) \Big\} = 0 \nonumber \\
%%(a^{\dagger} a^{\dagger})^j  (a^{\dagger} a^{\dagger} a^{\dagger})^{k-1} |0\rangle \nonumber \\
%%& + \sum_{\alpha=1}^{d^{n_F}} \sum_{t=1}^{n_B^{\alpha}} \sum_{2j+3k \le N_{cut}-n_B^{\alpha}} a^{\alpha}_{j,k}
%%(a^{\dagger} a^{\dagger})^j (a^{\dagger} a^{\dagger} a^{\dagger})^k G^t_{\alpha} |0 \rangle = 0.
%\end{align}
\begin{multline}
\sum_{\alpha=1}^{d^{n_F}} \sum_{2j+3k \le N_{cut}-n_B^{\alpha}} \Big\{ C^{\dagger}(n_B^{\alpha}, n_F, \alpha) \Big( \textbf{S}^{3k+3+n_B^{\alpha}}
\cdot \textbf{a}^{\alpha}_{j,k} \\ + \frac{3}{8}(k+1)(k+2) a^{\alpha}_{j-2,k+2} \Big) |j,k\rangle \\
 +  \sum_{t=1}^{n_B^{\alpha}} a^{\alpha}_{j,k} \big( k
\ \big[ G^t_{\alpha}, (a^{\dagger} a^{\dagger} a^{\dagger})\big] \ |j,k-1\rangle + (a^{\dagger} a^{\dagger})^j (a^{\dagger} a^{\dagger} a^{\dagger})^k G^t_{\alpha} |0 \rangle \big) \Big\} = 0  \\
\end{multline}
Remarkably, we recover a Laguerre recursion relation for the $a^{\alpha}_{j,k}$ coefficients with
an index shifted by $n_B^{\alpha}$ compared to the bosonic
case. The mixing terms produced by $\big[ G^t_{\alpha}, (a^{\dagger} a^{\dagger} a^{\dagger})\big]$
mix coefficients $a^{\alpha}_{j,k}$ with different $\alpha$.
Note that
$\big[ G^t_{\alpha}, (a^{\dagger} a^{\dagger} a^{\dagger})\big]$ cannot be proportional
to $C^{\dagger}(n_B^{\alpha}, n_F, \alpha)$ since it contains one additional bosonic creation
operator, whereas the part of $G^t_{\alpha}$ which does not annihilate the vacuum misses one
 creation operator compared to $C^{\dagger}(n_B^{\alpha}, n_F, \alpha)$.
Therefore, the two last terms of the recursion relation are true mixing terms.
We will investigate in details these recurrence relations in a particular case of the sector $n_F=1$ in the next section.

\section{Solutions in the $n_F=1$ sector}
\label{sec. Solutions in the $n_F=1$ sector}

In the sector $n_F=1$ we have two fermionic bricks, namely,
\begin{equation}
(f^{\dagger} a^{\dagger}) \equiv C^{\dagger}(1, 1, 1), \qquad (f^{\dagger} a^{\dagger} a^{\dagger}) \equiv C^{\dagger}(2, 1, 2).
\end{equation}
Hence we have,
\begin{align}
\begin{split}
\big[ (aa) , (f^{\dagger} a^{\dagger}) \big] & =  (f^{\dagger} a) \equiv G_1^1,  \\
\big[ (aa) , (f^{\dagger} a^{\dagger}a^{\dagger}) \big] &= (f^{\dagger} a^{\dagger} a) + (f^{\dagger} a a^{\dagger}) \equiv G_2^1 + G_2^2 .
\end{split}
\end{align}
Then,
\begin{align}
\begin{split}
\big[ G^1_1 , (a^{\dagger} a^{\dagger})^j (a^{\dagger} a^{\dagger} a^{\dagger})^k \big] &= \frac{3k}{2} (f^{\dagger} a^{\dagger} a^{\dagger}) (a^{\dagger} a^{\dagger})^j (a^{\dagger} a^{\dagger} a^{\dagger})^{k-1} \\
&\ + j \ (f^{\dagger} a^{\dagger}) (a^{\dagger} a^{\dagger})^{j-1}  (a^{\dagger} a^{\dagger} a^{\dagger})^k, \\
\big[ G_2^1, (a^{\dagger} a^{\dagger})^j (a^{\dagger} a^{\dagger} a^{\dagger})^k \big] &= \frac{k}{4} (f^{\dagger} a^{\dagger}) (a^{\dagger} a^{\dagger})^{j+1} (a^{\dagger} a^{\dagger} a^{\dagger})^{k-1}  \\
&\ + j \ (f^{\dagger} a^{\dagger} a^{\dagger}) (a^{\dagger} a^{\dagger})^{j-1}  (a^{\dagger} a^{\dagger} a^{\dagger})^k  \\
& = \big[ G_2^2, (a^{\dagger} a^{\dagger})^j (a^{\dagger} a^{\dagger} a^{\dagger})^k \big]  %&= \frac{k}{4} (f^{\dagger} a^{\dagger}) (a^{\dagger} a^{\dagger})^{j+1} (a^{\dagger} a^{\dagger} a^{\dagger})^{k-1} \nonumber \\
%&\ + j \ (f^{\dagger} a^{\dagger} a^{\dagger}) (a^{\dagger} a^{\dagger})^{j-1}  (a^{\dagger} a^{\dagger} a^{\dagger})^k. \nonumber
\end{split}
\end{align}
This gives the set of recurrences,
\begin{align}
\begin{split}
\textbf{S}^{3k+4}_{j} \cdot \textbf{a}^1_{j,k}
+ \frac{3}{8} (k+1)(k+2) \ a^1_{j-2,k+2} + \frac{k+1}{2} \ a^2_{j-1,k+1} = 0, \\
\label{eq. f1 recurrence relation}
\textbf{S}^{3k+5}_{j} \cdot \textbf{a}^2_{j,k}
+ \frac{3}{8} (k+1)(k+2) \ a^2_{j-2,k+2} + \frac{3(k+1)}{2} \ a^1_{j,k+1} = 0.
\end{split}
\end{align}

The next step is to decompose these recursion relation into even and odd parts. %
%In order to find the solutions of the above recursion relations we would like to use theorem
%\ref{th.szoste} of appendix \ref{sec.theorems}. Having this in mind, we will now reorganize
%these recursion relations.\\
%Notice that eqs.\eqref{eq. f1 recurrence relation} decouple into two separate systems.

The odd part contains equations
for $a^1_{j,k}$ with $k$ even and $a^2_{j,k}$ with $k$ odd. Therefore, we set $k=2m$ and $k=2m+1$ respectively
in the first and second recursion relation of eqs.\eqref{eq. f1 recurrence relation} and obtain
\begin{align}
\textbf{S}^{6m+4}_{j} \cdot \textbf{a}^1_{j,2m}
+ \frac{3}{8}(2m+1)(2m+2) \ a^1_{j-2,2m+2} + \frac{2m+1}{2} \ a^2_{j-1,2m+1} = 0, \nonumber \\
\textbf{S}^{6m+8}_{j} \cdot \textbf{a}^2_{j,2m+1}
+ \frac{3}{8}(2m+2)(2m+3) \ a^2_{j-2,2m+3} + 3(m+1) \ a^1_{j,2m+2} = 0. \nonumber
\end{align}
Since, the coefficients $a^1_{j,2m}$ appears always with an even second index, and the coefficients $a^2_{j,2m+1}$ appear always with
an odd second index, the superscript becomes redundant and can be omitted. Hence, we write $a^1_{j,2m} \rightarrow a_{j}^{2m}$
and $a^2_{j,2m+1} \rightarrow a_{j}^{2m+1}$ and get
\begin{align}
\begin{split}
 \textbf{S}^{6m+4}_{j} \cdot \textbf{a}_{j}^{2m}
 + \frac{2m+1}{2} \ a_{j-1}^{2m+1} + \frac{3}{8}(2m+1)(2m+2) \ a_{j-2}^{2m+2} = 0,  \\
 \label{eq. f1_recurrence relation odd}
 \textbf{S}^{6m+8}_{j} \cdot \textbf{a}_{j}^{2m+1}
 + 3(m+1) \ a_{j}^{2m+2} + \frac{3}{8}(2m+2)(2m+3) \ a_{j-2}^{2m+3} = 0.
\end{split}
\end{align}

%As far as the recursions from the second group are concerned, equations for
%$a^1_{j,k}$ with $k$ odd and $a^2_{j,k}$ with $k$ even.
We proceed similarly for the remaining set of equations. After setting \mbox{$k=2m+1$} and \mbox{$k=2m$} in the first
and second recursion relation of eqs.\eqref{eq. f1 recurrence relation} and changing the notation
$a^2_{j,2k} \rightarrow a^{2k}_j$ and $a^1_{j,2k+1} \rightarrow a^{2k+1}_j$, we have
\begin{align}
\begin{split}
 \textbf{S}^{6m+5}_{j} \cdot \textbf{a}_{j}^{2m}
 + \frac{3(2m+1)}{2} \ a_{j}^{2m+1} + \frac{3}{8} (2m+1)(2m+2) \ a_{j-2}^{2m+2} = 0,  \\
  \label{eq. f1_recurrence relation even}
 \textbf{S}^{6m+7}_{j} \cdot \textbf{a}_{j}^{2m+1}
 + (m+1) \ a_{j-1}^{2m+2} + \frac{3}{8}(2m+2)(2m+3) \ a_{j-2}^{2m+3} = 0.
\end{split}
\end{align}

Before writing the expressions for the solutions
of the above recurrence relations, we discuss the general structure of the solutions.

\subsection{Families of solutions}

In a similar manner as we did in the bosonic case, it turns out to be possible to
classify the solutions in the $n_F=1$ sector into families.
We have just argued that parity is a good quantum number.
Furthermore, the set of solutions of the same parity can be divided into two distinct sets of families.

We adopt the following notation:
\begin{itemize}
\item a solution belongs to the family $g^1_{\kappa}$ if it is composed of basis states with an \emph{odd} number of quanta and
the basis state with the maximal number of cubic bricks equal to $2{\kappa}$, is proportional to the $(f^{\dagger}a^{\dagger})$ brick,
\item a solution belongs to the family $g^2_{\kappa}$ if it is composed of basis states with an \emph{odd} number of quanta and
the basis state with the maximal number of cubic bricks equal to $2{\kappa}+1$, is proportional to the $(f^{\dagger}a^{\dagger}a^{\dagger})$ brick.
%$a^1_{j,k} \equiv 0, k > 2m + 1$
%and $a^1_{j,k} \ne 0, k \le 2m + 1$ or a solution belongs to the set $f^2_m$ if $a^2_{j,k} \equiv 0, k > 2m$
%and $a^2_{j,k} \ne 0, k \le 2m$,
%\item a solution belongs to the set $g^1_m$ if $a^1_{j,k} \equiv 0, k > 2m$
%and $a^1_{j,k} \ne 0, k \le 2m$ or a solution belongs to the set $g^2_m$ if $a^2_{j,k} \equiv 0, k > 2m + 1$
%and $a^2_{j,k} \ne 0, k \le 2m + 1$
\end{itemize}
The recursion relation for the families $g^1_{\kappa}$ and $g^2_{\kappa}$ are given by eqs.\eqref{eq. f1_recurrence relation odd}.
\begin{itemize}
\item a solution belongs to the family $f^1_{\kappa}$ if it is composed of basis states with an \emph{even} number of quanta and
the basis state with the maximal number of cubic bricks equal to $2{\kappa}+1$, is proportional to the $(f^{\dagger}a^{\dagger})$ brick,
\item a solution belongs to the family $f^2_{\kappa}$ if it is composed of basis states with an \emph{even} number of quanta and
the basis state with the maximal number of cubic bricks equal to $2{\kappa}$, is proportional to the $(f^{\dagger}a^{\dagger}a^{\dagger})$ brick.
\end{itemize}
The recursion relation for the families $f^1_{\kappa}$ and $f^2_{\kappa}$ are given by eqs.\eqref{eq. f1_recurrence relation even}.\\

One can easily calculate the number of solutions within each family,
by counting the number of basis states fulfilling appropriate conditions.
Hence, in analogy to the bosonic case we define $d^1_{\kappa}(N_{cut})$ and $d^2_{\kappa}(N_{cut})$
as the numbers of solutions in the families $f^1_{\kappa}$ and $f^2_{\kappa}$ at cut-off $N_{cut}$.
Similarly, $d'^1_{\kappa}(N_{cut})$ and $d'^2_{\kappa}(N_{cut})$ denote the numbers
of solutions in the families $g^1_{\kappa}$ and $g^2_{\kappa}$ at cut-off $N_{cut}$.
They are given by
\begin{align}
d'^1_{\kappa}(N_{cut}) = \Big\lfloor \frac{1}{2} \big( N_{cut}-6\kappa-1) \big) \Big\rfloor + 1,  & \qquad
d'^2_{\kappa}(N_{cut}) = \Big\lfloor \frac{1}{2} \big( N_{cut}-6\kappa-5) \big) \Big\rfloor + 1, \\
d^1_{\kappa}(N_{cut}) = \Big\lfloor \frac{1}{2} \big( N_{cut}-6\kappa-4) \big) \Big\rfloor + 1, & \qquad
d^2_{\kappa}(N_{cut}) = \Big\lfloor \frac{1}{2} \big( N_{cut}-6\kappa-2) \big) \Big\rfloor + 1.
\end{align}
Obviously, for a given cut-off $N_{cut}$ there will be in total $\lfloor \frac{N_{cut}-1}{3} \rfloor + \lfloor \frac{N_{cut}-2}{3} \rfloor$ families of solutions.

\subsection{Generic solution and the complete set of solutions at finite cut-off}

In this section we use theorem \ref{th.szoste} from appendix \ref{sec.theorems} to write down the solutions to the
recursion relations eqs.\eqref{eq. f1_recurrence relation odd} and eqs.\eqref{eq. f1_recurrence relation even}.
We treat a generic cases
of solutions belonging to the families $g^1_{\kappa}$ and $g^2_{\kappa}$ and
then to the families $f^1_{\kappa}$ and $f^2_{\kappa}$.
Eventually, we describe the complete set of solutions at finite cut-off.

\subsubsection{Solutions with odd number of quanta}

The quantization condition for the family $g^1_{\kappa}$ reads
\begin{equation}
L_{d'^1_{\kappa}}^{6\kappa+4}(2E_m) = 0, \qquad 1 \le m \le d_{\kappa}'^1,
\end{equation}
and the corresponding eigenstates are given by
\begin{multline}
|E_m,\kappa,1\rangle_{odd} = (2E_m)^{3\kappa+\frac{1}{2}} e^{-E_m} \sum_{n=0}^{d'^1_{\kappa}-1} L_n^{6\kappa+4}(2E_m)
\Big\{ (f^{\dagger}a^{\dagger})|n,2\kappa\rangle + \\
+ \sum_{p=1}^{\kappa} \Big( \Gamma^{odd}_{2\kappa, 2\kappa-2p+1} (f^{\dagger}a^{\dagger}a^{\dagger})|n+3p-2,2\kappa-2p+1\rangle + \\
+ \Gamma^{odd}_{2\kappa,2\kappa-2p} (f^{\dagger}a^{\dagger})|n+3p,2\kappa-2p\rangle \Big) \Big\}.
\end{multline}
Similarly, the quantization condition for the family $g^2_{\kappa}$ reads
\begin{equation}
L_{d'^2_{\kappa}}^{6\kappa+8}(2E_m) = 0, \qquad 1 \le m \le d_{\kappa}'^2,
\end{equation}
with the eigenstates if the form,
\begin{multline}
|E_m,\kappa,2\rangle_{odd} = (2E_m)^{3\kappa+\frac{5}{2}} e^{-E_m} \sum_{n=0}^{d'^2_{\kappa}-1} L_n^{6\kappa+8}(2E_m)
\Big\{ (f^{\dagger}a^{\dagger}a^{\dagger})|n,2\kappa+1\rangle + \\
+ \sum_{p=1}^{\kappa} \Big( \Gamma^{odd}_{2\kappa+1,2\kappa-2p+2} (f^{\dagger}a^{\dagger})|n+3p-1,2\kappa-2p+2\rangle + \\
+ \Gamma^{odd}_{2\kappa+1,2\kappa-2p+1} (f^{\dagger}a^{\dagger}a^{\dagger})|n+3p,2\kappa-2p+1\rangle \Big) \Big\}.
\end{multline}
$\Gamma^{odd}_{m,p}$ are given by recursive relations in theorem \ref{th.szoste} presented in appendix \ref{sec.theorems}.

\subsubsection{Solutions with even number of quanta}

In this case, the quantization condition for the family $f^1_{\kappa}$ reads
\begin{equation}
L_{d^1_{\kappa}}^{6\kappa+7}(2E_m) = 0, \qquad 1 \le m \le d^1_{\kappa},
\end{equation}
and the corresponding eigenstate is of the form
\begin{multline}
|E_m,\kappa,1\rangle_{even} = (2E_m)^{3\kappa+2} e^{-E_m} \sum_{n=0}^{d^1_{\kappa}-1} L_n^{6\kappa+7}(2E_m)
\Big\{ (f^{\dagger}a^{\dagger})|n,2\kappa +1\rangle + \\
+ \sum_{p=1}^{\kappa} \Big( \Gamma^{even}_{2\kappa+1,2\kappa-2p} (f^{\dagger}a^{\dagger}a^{\dagger})|n+3p+1,2\kappa-2p\rangle + \\
+ \Gamma^{even}_{2\kappa+1,2\kappa-2p+1} (f^{\dagger}a^{\dagger})|n+3p,2\kappa-2p+1\rangle \Big) \Big\}.
\end{multline}
The quantization condition for the family $f^2_{\kappa}$ is given by
\begin{equation}
L_{d^2_{\kappa}}^{6\kappa+5}(2E_m) = 0, \qquad 1 \le m \le d^2_{\kappa},
\end{equation}
and the eigenstate reads
\begin{multline}
|E_m,\kappa, 2\rangle_{even} = (2E_m)^{3\kappa+1} e^{-E_m} \sum_{n=0}^{d^2_{\kappa}-1} L_n^{6\kappa+5}(2E_m)
\Big\{ (f^{\dagger}a^{\dagger}a^{\dagger})|n,2\kappa\rangle + \\
+ \sum_{p=1}^{\kappa} \Big( \Gamma^{even}_{2\kappa,2\kappa-2p+1} (f^{\dagger}a^{\dagger})|n+3p-1,2\kappa-2p+1\rangle + \\
+ \Gamma^{even}_{2\kappa,2\kappa-2p} (f^{\dagger}a^{\dagger}a^{\dagger})|n+3p,2\kappa-2p\rangle \Big) \Big\}.
\end{multline}
Again, $\Gamma^{even}_{m,p}$ are as given in theorem \ref{th.szoste}.

\subsubsection{Complete set of solutions at finite cut-off}

In a similar manner as in the case of the bosonic sector the complete set of solutions at a given cut-off $N_{cut}$,
$\big\{ | E \rangle \big\}_{N_{cut}}$ is
given by the union of all solutions from nonempty families.
%\begin{equation}
%\big\{ | E \rangle \big\}_{N_{cut}} =
%\bigcup_{m=0}^{\big\lfloor \frac{1}{6} \big( N_{cut} -4 \big) \big\rfloor} f^1_m \
%\cup \ \bigcup_{m=0}^{\big\lfloor \frac{1}{6} \big( N_{cut} -2 \big) \big\rfloor} f^2_m
%\cup \bigcup_{m=0}^{\big\lfloor \frac{1}{6} \big(N_{cut} -1 \big) \big\rfloor} g^1_m \
%\cup \ \bigcup_{m=0}^{\big\lfloor \frac{1}{6} \big( N_{cut} -5 \big) \big\rfloor} g^2_m.
%\label{eq. complete eigenstates f1}
%\end{equation}
We write explicitly only the $d'^1_0(N_{cut})$ odd and $d_0^2(N_{cut})$ even simplest eigenstates with no mixing, namely,
\begin{equation}
|E_m,0,1 \rangle_{odd} = (2E_m)^{\frac{1}{2}} e^{-E_m} \sum_{n=0}^{d'^1_0-1} L_n^{4}(2E_m) (f^{\dagger}a^{\dagger}) |n,0\rangle,
\label{eq. fermionic 1}
\end{equation}
with
\begin{equation}
L^4_{d'^1_0}(2E_m)=0, \qquad 1 \le m \le d'^1_0, \nonumber
\end{equation}
and
\begin{equation}
|E_m,0,2 \rangle_{even} = (2E_m) e^{-E_m} \sum_{n=0}^{d^2_0-1} L_n^{5}(2E_m) (f^{\dagger}a^{\dagger}a^{\dagger}) |n,0\rangle,
\label{eq. fermionic 2}
\end{equation}
with
\begin{equation}
L^5_{d_0^2}(2E_m)=0, \qquad 1 \le m \le d_0^2. \nonumber
\end{equation}

%,
%can be written as the union of sets
%\begin{align}
%\big\{ E \big\} &=
%\Bigg( \bigcup_{m=0}^{\big\lfloor \frac{1}{6} \big( N_{cut}-4\big) \big\rfloor} \Big\{ L_{d^1_m}^{6m+5}(2E) = 0\Big\} \
%\cup \ \bigcup_{m=0}^{\big\lfloor \frac{1}{6} \big( N_{cut} -2 \big) \big\rfloor} \Big\{ L_{d^2_m}^{6m+7}(2E) = 0\Big\} \Bigg)\nonumber \\
% &\cup \ \Bigg( \bigcup_{m=0}^{\big\lfloor \frac{1}{6} \big(N_{cut}-1\big) \big\rfloor} \Big\{ L_{d'^1_m}^{6m+4}(2E) = 0\Big\} \
% \cup \ \bigcup_{m=0}^{\big\lfloor \frac{1}{6} \big( N_{cut} -5 \big) \big\rfloor} \Big\{ L_{d'^2_m}^{6m+8}(2E) = 0\Big\} \Bigg).
%\label{eq. complete eigenenergies f1}
%\end{align}
The spectrum, i.e. the set of all values of the $E$ parameter, $\big\{ E \big\}_{N_{cut}}$,
for which a nonzero eigenstate exists, can be written in a compact form using the polynomial $\Theta^{n_F=1}_{N_{cut}}(E)$. The latter can
be deduced to be of the form
\begin{multline}
\Theta_{N_{cut}}^{n_F=1}(E) = \Big[\prod_{j=0}^{\lfloor \frac{1}{3} ( N_{cut}-1 ) \rfloor} L_{\lfloor \frac{1}{2} ( N_{cut}-3j-1 ) \rfloor + 1}^{3j+4}(2E) \Big] \ \times \\
 \times \ \Big[ \prod_{k=0}^{\lfloor \frac{1}{3} ( N_{cut}-2 ) \rfloor} L_{\lfloor \frac{1}{2} ( N_{cut}-3k-2 ) \rfloor + 1}^{3k+5}(2E) \Big].
\label{eq. complete eigenenergies fermionic}
\end{multline}
Hence, again, the expression eq.\eqref{eq. complete eigenenergies fermionic}
provides a closed formula for the spectrum of Hamiltonian in the fermionic sector for
any finite cut-off.

\subsection{Continuum solutions}

The argument described in section \ref{subsec. Continuum solutions} applies as well to the case
of recursion relations eqs.\eqref{eq. f1 recurrence relation}. Therefore,
one can obtain exact, continuum solutions from the solutions described in the preceding section, by simply extending
the sums to infinity, i.e. solutions eqs.\eqref{eq. fermionic 1} and \eqref{eq. fermionic 2} become
\begin{align}
\begin{split}
|E,0,1 \rangle_{odd} = (2E)^{\frac{1}{2}} e^{-E} \sum_{n=0}^{\infty} L_n^{4}(2E) (f^{\dagger}a^{\dagger}) |n,0\rangle, \\
|E,0,2 \rangle_{even} = (2E) e^{-E} \sum_{n=0}^{\infty} L_n^{5}(2E) (f^{\dagger}a^{\dagger}a^{\dagger}) |n,0\rangle.
\end{split}
\end{align}
The energy quantization conditions in the continuum limit allow energies that lie on
the whole real, positive axis. %The question of zero-energy states is more subtle and will be dealt elsewhere.

\section{Spectra in all fermionic sectors}
\label{sec. Spectra in all fermionic sectors}

In section \ref{sec. General recurrence relation in the fermionic sectors of the $SU(3)$ model}
we derived a general recurrence relation valid in all fermionic sectors.
Remarkably, it has the form of a set of generalized Laguerre recurrence relations coupled by a number of mixing terms.
However, for any finite cut-off $N_{cut}$ there always exists one homogeneous Laguerre equation for which the mixing vanishes since it involves
terms with a number of quanta beyond $N_{cut}$. By solving these homogeneous equations and then proceeding
with the remaining equations (in analogy to the proofs of theorems \ref{th.trzecie} and \ref{th.szoste}),
it is possible to obtain solutions of all recursions in all fermionic sectors.
In particular, the corollary \ref{th.piate}, which justifies the existence of a compact polynomial $\Theta_{N_{cut}}^{n_F}(E)$
which roots corresponds to the eigenenergies of the Hamiltonian, can be extended to all sectors.

The polynomial $\Theta_{N_{cut}}^{n_F}(E)$ is equal to the product of quantization conditions of all nonempty families
of solutions at any finite cut-off $N_{cut}$ in a given fermionic sector.
Each quantization condition is characterized by two integers, $n$ and $\gamma$, and has the form $L_{n}^{\gamma}(2E)=0$.
The index $\gamma$ is given by the general recursion relation, i.e. $\gamma = 3k+3+n_B^{\alpha}(n_F)$, where for
each $i$, the index $k$ enumerates the families of a specific type and
is bounded by $0 \le k \le \lfloor \frac{1}{3} \big( N_{cut}-n_B^{\alpha}(n_F) \big) \rfloor$.
The index $n$ corresponds to the number of solutions within each family and it is equal to
$n = \Big\lfloor \frac{1}{2} \big( N_{cut}-3k-n_B^{\alpha}(n_F) \big) \Big\rfloor + 1$.
Hence, the general polynomial can be written as
\begin{equation}
\Theta_{N_{cut}}^{n_F}(E) = \prod_{\alpha=1}^{d^{n_F}} \Bigg( \prod_{k=0}^{\lfloor \frac{1}{3} ( N_{cut}-n_B^{\alpha}(n_F) ) \rfloor} L_{\lfloor \frac{1}{2} ( N_{cut}-3k-n_B^{\alpha}(n_F) ) \rfloor + 1}^{3k+3+n_B^{\alpha}(n_F)}(2E) \Bigg).
\end{equation}

Therefore, in order to be able to determine the spectra the numbers $n_B^{\alpha}(n_F)$ must be known.
In the case of the $SU(3)$ model they can be simply read from the table \ref{tab. su3_fermionic_bricks}.
The above prescription for the spectra was crosschecked with independent numerical calculations
up to cut-off $N_{cut} \le 40$, in all fermionic sectors $0 \le n_F \le 4$. Both
computations yielded exactly the same numbers.

\section{Conclusions}
\label{sec. Conclusions}

In this article we have studied the model of $D=2$, Supersymmetric Yang-Mills Quantum Mechanics with a $SU(3)$ gauge group. We have proposed
a method for solving this system. It applies as well to the bosonic sector as to the
fermionic ones. We have derived explicit formulas for the eigenstates in the
bosonic sector and in the sector with one fermion by exactly
solving a recursion relation involving mixing terms. Such mixing terms are a novel, characteristic feature
of the $SU(N>2)$, SYMQM models. Moreover, we obtained expressions giving the energy spectrum in all fermionic sectors which are parameterized by the
cut-off. Hence, on one hand, one has an analytic confirmation of the correctness of numerical calculations, which are always done at finite cut-off.
On the
other hand, one can easily and in a precise way extrapolate these results to the continuum limit and
therefore obtain expressions for the exact spectrum and corresponding eigenstates.

These results can prove to be valuable in studies of several further problems.
A generalization to solutions of the models with $SU(N)$ gauge groups with $N>3$ is possible
not only in the bosonic sector but in fermionic sectors as well.
%As an example, the numbers $n_B^{\alpha}$ for the gauge group
%$SU(4)$ were provided in \cite{doktorat_macka}.
The results for general $N$ may be very useful in the investigation of the large-$N$ limit.
Furthermore, in the search of solutions of more physically interesting, i.e. $D=4$ or $D=10$ dimensional models, one can
try to derive and solve a similar recursion relations.

\section*{Acknowledgements}

The Author would like to thank J. Wosiek for many discussions on the subject of this paper.

\appendix

\section{Derivation of recursion relation eq. \eqref{eq.su3.recurrence.rel}}
\label{app. 2}

We have,
\begin{equation}
H = (a^{\dagger}a) + 2 -\frac{1}{2}(aa) -\frac{1}{2}(a^{\dagger}a^{\dagger}),
\end{equation}
and
\begin{equation}
|E\rangle = \sum_{i,j=0}^{\infty} a_{i,j} |i,j\rangle = \sum_{i,j=0}^{\infty} a_{i,j} (a^{\dagger} a^{\dagger})^i(a^{\dagger} a^{\dagger} a^{\dagger})^j |0\rangle.
\end{equation}
One can calculate that,
\begin{align}
(a^{\dagger}a) (a^{\dagger} a^{\dagger})^i (a^{\dagger}a^{\dagger}a^{\dagger})^j |0\rangle &= \big( i+\frac{3}{2}j \big) (a^{\dagger} a^{\dagger})^i (a^{\dagger}a^{\dagger}a^{\dagger})^j |0\rangle, \\
(a^{\dagger}a^{\dagger})(a^{\dagger} a^{\dagger})^i (a^{\dagger} a^{\dagger}a^{\dagger})^j |0\rangle &= (a^{\dagger} a^{\dagger})^{i+1} (a^{\dagger}a^{\dagger}a^{\dagger})^j |0\rangle, \\
(aa)(a^{\dagger} a^{\dagger})^i(a^{\dagger} a^{\dagger} a^{\dagger})^j  |0\rangle &= i(i+3j+3)(a^{\dagger}a^{\dagger})^{i-1} (a^{\dagger}a^{\dagger}a^{\dagger})^j |0\rangle \nonumber \\
& + \frac{3}{8} j (j-1) (a^{\dagger}a^{\dagger})^{i+2} (a^{\dagger}a^{\dagger}a^{\dagger})^{j-2} |0\rangle,
\end{align}
where for the last equality we have used the following commutators,
\begin{align}
\big[ (aa), (a^{\dagger} a^{\dagger})^i \big]
%&= \sum_{k=0}^{i-1} (a^{\dagger}a^{\dagger})^{i-1-k} \big[(aa), (a^{\dagger}a^{\dagger}) \big] (a^{\dagger} a^{\dagger})^k \nonumber \\
%&= \sum_{k=0}^{i-1} (a^{\dagger}a^{\dagger})^{i-1-k} \big( 2(a^{\dagger}a) + 4 \big) (a^{\dagger} a^{\dagger})^k \nonumber \\
%&= \sum_{k=0}^{i-1} (a^{\dagger}a^{\dagger})^{i-1-k} \big( 2k + 4 \big) (a^{\dagger} a^{\dagger})^k + 2 i (a^{\dagger}a^{\dagger})^{i-1} (a^{\dagger} a)\nonumber \\
&= i(i+3) (a^{\dagger}a^{\dagger})^{i-1} + 2 i (a^{\dagger}a^{\dagger})^{i-1} (a^{\dagger} a),
\end{align}
\begin{align}
\big[ (aa), (a^{\dagger} a^{\dagger}a^{\dagger})^i \big]
%&= \sum_{k=0}^{i-1} (a^{\dagger}a^{\dagger}a^{\dagger})^{i-1-k} \big[(aa), (a^{\dagger}a^{\dagger}a^{\dagger}) \big] (a^{\dagger}a^{\dagger}a^{\dagger})^k \nonumber \\
%&= 3\sum_{k=0}^{i-1} (a^{\dagger}a^{\dagger}a^{\dagger})^{i-1-k}(a^{\dagger}a^{\dagger}a)(a^{\dagger}a^{\dagger}a^{\dagger})^k \nonumber \\
%&= \frac{3}{4} \sum_{k=0}^{i-1} k \ (a^{\dagger}a^{\dagger}a^{\dagger})^{i-2} (a^{\dagger}a^{\dagger})^2 + 3 i \ (a^{\dagger}a^{\dagger}a^{\dagger})^{i-1} (a^{\dagger}a^{\dagger} a)\nonumber \\
&= \frac{3}{8} i(i-1) \ (a^{\dagger}a^{\dagger}a^{\dagger})^{i-2} (a^{\dagger}a^{\dagger})^2 + 3 i \ (a^{\dagger}a^{\dagger}a^{\dagger})^{i-1} (a^{\dagger}a^{\dagger} a),
\end{align}
and
\begin{align}
\big[ (a^{\dagger} a^{\dagger}a), (a^{\dagger} a^{\dagger}a^{\dagger})^k \big] &= \frac{1}{4} k (a^{\dagger} a^{\dagger}a^{\dagger})^{k-1} (a^{\dagger} a^{\dagger})^2
\end{align}
Thus,
\begin{multline}
H |i,j\rangle
= \big( i +\frac{3}{2}j + 2 \big)|i,j\rangle
-\frac{1}{2} |i+1,j\rangle + \\ - \frac{1}{2} i \big(i+ 3 j +3  \big) |i-1,j\rangle
- \frac{3}{16} j(j-1)  |i+2,j-1\rangle
\end{multline}
Therefore, grouping the coefficients in front of each basis state yields the recursion relation
%\begin{align}
%H |E\rangle &= \sum_{i,j=0}^{\infty} a_{i,j} \Big(
%\big( i +\frac{3}{2}j + 2 \big)(a^{\dagger} a^{\dagger})^i(a^{\dagger} a^{\dagger}a^{\dagger})^j |0\rangle
%-\frac{1}{2} (a^{\dagger} a^{\dagger})^{i+1}(a^{\dagger} a^{\dagger}a^{\dagger})^j |0\rangle \nonumber \\
%&- \frac{1}{2} i \big(i+ 3 j +3  \big) (a^{\dagger}a^{\dagger})^{i-1} (a^{\dagger} a^{\dagger}a^{\dagger})^j |0\rangle %\nonumber \\
%- \frac{3}{16} j(j-1)  (a^{\dagger}a^{\dagger})^{i+2} (a^{\dagger}a^{\dagger}a^{\dagger})^{j-2} |0\rangle
% \Big) \nonumber \\
%&= \sum_{i,j=0}^{\infty} \Big( a_{i,j} \big( i + \frac{3}{2}j + 2 \big) -\frac{1}{2}a_{i-1,j} - \frac{1}{2} (i+1)\big( i + 3 j + 4 \big)a_{i+1,j} \nonumber \\
%&- \frac{3}{16} (j+1)(j+2) a_{i-2,j+2} \Big) (a^{\dagger} a^{\dagger})^i(a^{\dagger} a^{\dagger}a^{\dagger})^j |0\rangle
%\end{align}
%So,
\begin{multline}
a_{i-1,j} - \big( 2i + 3j + 4 -2E \big) a_{i,j} + (i+1)\big( i + 3 j + 4 \big)a_{i+1,j} + \\+ \frac{3}{8} (j+1)(j+2)a_{i-2,j+2}  = 0.
\end{multline}

\section{Solutions of the cut recursion relations}
\label{sec.theorems}
%\label{app. theorems}

In this section we solve cut recursion relations for the generalized Laguerre polynomials.
We construct the theorems in the order of increasing complexity.
We start by providing the solution to a single cut recursion relation (lemma \ref{lem. pierwszy}).
Next, we consider respectively, the problem of a single inhomogeneous cut recursion relation (lemma \ref{th.pierwsze}),
the problem of two coupled recursion relations (lemma \ref{th.drugie}),
and the problem of a set of coupled recursion relations (theorem \ref{th.trzecie}).
Eventually, we solve a problem of a set of twofold coupled recursion relations (theorem \ref{th.szoste}) and formulate a
corollary (corollary \ref{th.piate}), which can be used to get the spectra of the system in all fermionic sectors.

\subsection{Some lemmas}

Let us start by considering the following lemma,
\begin{lemma}
\label{lem. pierwszy}
Consider a set of $\kappa+1$ generalized Laguerre equations for the coefficients $a_0, \dots, a_{\kappa}$,
\begin{equation}
a_{j-1}(x) - \big( 2j + \alpha + 1 - x \big) a_{j}(x) + (j+1)(j+\alpha+1) a_{j+1}(x) = 0.
\end{equation}
Then, there exist $\kappa+1$ nontrivial solutions, denoted by $(a_j)^0, \dots, (a_j)^{\kappa}$, of the form
\begin{equation}
(a_j)^i = a_0 \Gamma(\alpha+1) L_j^{\alpha}(x_i),
\end{equation}
where $x_i$ are the solutions of the equation
\begin{equation}
L_{\kappa+1}^{\alpha}(x_i) = 0, \qquad i=0, \dots, \kappa.
\end{equation}
\end{lemma}
\begin{proof}
Let us write this cut set of generalized Laguerre equations in an explicit form as
\begin{align}
&-(\alpha + 1 - x) a_0 + (\alpha + 1) a_1 = 0, \nonumber \\
a_0 &- (\alpha + 3 - x) a_1 +2 (\alpha +2) a_2 = 0, \nonumber\\
a_1 &- (\alpha + 5 - x) a_2 +3(\alpha+3) a_3 = 0, \nonumber\\
a_2 &- (\alpha + 7 - x) a_3 +4(\alpha+4) a_4 = 0, \nonumber\\
\vdots \nonumber \\
a_{\kappa-4} &- (\alpha + 2\kappa-5 - x) a_{\kappa-3} + (\kappa-2)(\alpha + \kappa-2) a_{\kappa-2} = 0, \nonumber\\
a_{\kappa-3} &- (\alpha + 2\kappa-3 - x) a_{\kappa-2} + (\kappa-1) (\alpha + \kappa-1) a_{\kappa-1} = 0, \nonumber\\
a_{\kappa-2} &- (\alpha + 2\kappa-1 - x) a_{\kappa-1} + \kappa(\alpha+\kappa) a_{\kappa} = 0, \nonumber \\
a_{\kappa-1} &- (\alpha + 2\kappa+1 - x) a_{\kappa} = 0. \nonumber
\end{align}
Now, one can express $a_1$ in terms of $a_0$ using the first equation, then express $a_2$ in terms of $a_0$ using the second equation. Proceeding in this way
for all but the last equation, enables us
to rewrite the above set of equations as
\begin{align}
a_1 &= \frac{\alpha + 1 - x}{\alpha+1} a_0  = a_0 \ \Gamma(\alpha+1) L^{\alpha}_1(x), \nonumber \\
a_2 &= \frac{((\alpha + 3 - x)\frac{\alpha + 1 - x}{\alpha+1}-1)}{2 (\alpha +2)}a_0 = a_0 \ \Gamma(\alpha+1) L^{\alpha}_2(x), \nonumber\\
a_3 &= a_0 \ \Gamma(\alpha+1) L^{\alpha}_3(x), \nonumber\\
a_4 &= a_0 \ \Gamma(\alpha+1) L^{\alpha}_4(x), \nonumber\\
\vdots \nonumber \\
a_{\kappa-2} &= a_0 \ \Gamma(\alpha+1) L^{\alpha}_{\kappa-2}(x), \nonumber\\
a_{\kappa-1} &= a_0 \ \Gamma(\alpha+1) L^{\alpha}_{\kappa-1}(x), \nonumber\\
a_{\kappa} &= a_0 \ \Gamma(\alpha+1) L^{\alpha}_{\kappa}(x). \nonumber
\end{align}
The last equation reads,
\begin{equation}
a_0 L^{\alpha}_{\kappa-1}(x) - a_0 (\alpha + 2\kappa+1 - x) L^{\alpha}_{\kappa}(x) = 0.
\end{equation}
Using eq.\eqref{eq.rescaled.recurrence.relation} it can be transformed into
\begin{equation}
a_0 L^{\alpha}_{\kappa+1}(x) = 0
\end{equation}
Thus, if $x$ is tuned to be a root of $L^{\alpha}_{\kappa+1}$, then $a_0$
can be arbitrary and the system admits a nonzero solution.
On the contrary, if $x$ is not a root of
$L^{\alpha}_{\kappa+1}$, then $a_0$ must vanish. Hence all coefficients $a_j$ vanish too.

Eventually, the equation $L^{\alpha}_{\kappa+1}(x) = 0$ has $\kappa+1$ zeros, so the system
has $\kappa+1$ independent solutions. We denote them by
\begin{equation}
(a_j)^0, \dots, (a_j)^{\kappa}
\end{equation}
with
\begin{equation}
(a_j)^i = a_0 \Gamma(\alpha+1) L_j^{\alpha}(x_i).
\end{equation}
\end{proof}

We now prove two lemmas giving the solutions of an inhomogeneous Laguerre recursion relations.

\begin{lemma}
\label{th.pierwsze}
Consider a set of $\kappa+1$, inhomogeneous, generalized Laguerre equations of the following form,
\begin{equation}
a_{j-1}(x) - \big( 2j + \alpha + 1 - x \big) a_{j}(x) + (j+1)(j+\alpha+1) a_{j+1}(x) = \chi L^{\beta}_{j+q}(x),
\label{eq. theorem 1}
\end{equation}
with $\chi \ne 0$, $\alpha$ and $\beta$ such that $\frac{1}{2}(\beta-\alpha)$ is a
positive integer and $x$ not being a root of the polynomial $L_{\kappa+1}^{\alpha}(x)$.

Then, there exist specific values of $\alpha$, $\beta$ and $q$, for which the system admits solutions
with $a_j$ nonvanishing.

\end{lemma}
\begin{proof}
Equations \eqref{eq. theorem 1} form a set of
$\kappa+1$ inhomogeneous equations with $\kappa+1$ variables $a_j(x)$. Its determinant is equal to $L_{\kappa+1}^{\alpha}(x)$ and, by assumption,
does not vanish. Therefore, the system admits a unique solution. We will construct it in the following way. We assume that $a_j(x)$ is proportional
to $L^{\beta}_{j-p}$ with some proportionality factor $\gamma \ne 0$ and $p$ some integer to be fixed latter,
\begin{equation}
a_j(x) = \gamma L^{\beta}_{j-p}(x).
\label{eq. zalozenie}
\end{equation}
Then, the recursion relation for $a_j(x)$ takes the form
\begin{align}
L^{\beta}_{j-p-1}(x) &- \big( 2j + \alpha + 1 - x \big) L^{\beta}_{j-p}(x) + \nonumber \\
&+ (j+1)(j+\alpha+1) L^{\beta}_{j-p+1}(x) + \frac{\chi}{\gamma} L^{\beta}_{j+q}(x) = 0.
\label{eq. 1}
\end{align}
The general recursion relation for $L^{\beta}_j$ is
\begin{align}
L^{\beta}_{j-1}(x) - \big( 2j + \beta + 1 - x \big) L^{\beta}_{i}(x) + (j+1)(j+\beta+1) L^{\beta}_{j+1}(x) = 0.
\label{eq. 2}
\end{align}
By assumption we have that $\frac{1}{2}(\beta-\alpha) = k$ with $k$ integer, $k>0$. Thus, for $j \ge k$ we can
shift the index $j \rightarrow j + k$ in eq.\eqref{eq. 1}. After setting $p = k$, we obtain,
\begin{align}
L^{\beta}_{j-1}(x) &- \big( 2j + \beta + 1 - x \big) L^{\beta}_{j}(x) + \nonumber \\
&+ (j+k+1)(j+k+\alpha+1) L^{\beta}_{j+1}(x) + \frac{\chi}{\gamma} L^{\beta}_{j+k+q}(x) = 0,
\label{eq. 3}
\end{align}
which can be considerably reduced using eq. \eqref{eq. 2}. Thus, eventually,
\begin{align}
- \frac{1}{4} (\beta^2 - \alpha^2) L^{\beta}_{j+1}(x) + \frac{\chi}{\gamma} L^{\beta}_{j+k+q}(x) = 0.
\label{eq. 4}
\end{align}
Eq. \eqref{eq. 4} must be satisfied for any value of the $x$ parameter.
This can only happen when
\begin{align}
q&=1-k = 1 - \frac{1}{2}(\beta-\alpha), \\
\gamma &= \frac{4 \chi}{\beta^2 - \alpha^2}.
\end{align}
It follows from eq. \eqref{eq. zalozenie} that the coefficients $a_j(x)$ for $j < k$ must vanish.

Summarizing, we showed that there exist a nonvanishing solution to the initial set of equations, given by
\begin{align}
a_j(x_i) &= 0, \qquad j=0, \dots, k-1 \\
a_j(x_i) &= \gamma L_{j-k}^{\beta}(x_i), \qquad j=k, \dots, \kappa,
\end{align}
where $k = \frac{1}{2}(\beta-\alpha)$. Note that this solution has no more the freedom of arbitrary constant factor.
\end{proof}

\begin{lemma}
\label{th.drugie}
Consider two cut, generalized Laguerre recursion relations, one of them supplemented with
a mixing term, $\chi$, of the following form,
\begin{align}
a_{j-1}(x) - \big( 2j + \alpha + 1 - x \big) a_{j}(x) + (j+1)(j+\alpha+1) a_{j+1}(x) - \chi b_{j+q}(x) = 0, \nonumber \\
b_{j-1}(x) - \big( 2j + \beta + 1 - x \big) b_{j}(x) + (j+1)(j+\beta+1) b_{j+1}(x) = 0,
\end{align}
and let us assume that $\chi \ne 0$, $\frac{1}{2}(\beta-\alpha) = k$ with $k$ a positive integer and $q \le 0$.

Then, for some specific values of the parameter $x$ the system admits solutions
with both $a_j(x)$ and $b_j(x)$ nonvanishing.
\end{lemma}
\begin{proof}
The set of
equations for $b_j(x)$ can be brought the form of lemma \ref{lem. pierwszy}. Therefore, for the
cut-off $\kappa$, it admits $\kappa+1$ nonzero solutions.
Consistency requires to assume that there are $\kappa+2-q$ coefficients $a_j(x)$.
Considering the possible values of the $x$ parameter we have three cases:
\begin{itemize}
\item $L_{\kappa+2-q}^{\alpha}(x) \ne 0$ and $L_{\kappa+1}^{\beta}(x) \ne 0$, then
    neither, $a_j(x)$ nor $b_j(x)$ admit a nonvanishing solutions.
\item $L_{\kappa+2-q}^{\alpha}(x) = 0$ and $L_{\kappa+1}^{\beta}(x) \ne 0$, then
    $b_j(x)$ vanish, yielding the homogeneous set of equations for $a_j(x)$. Because $L_{\kappa+2-q}^{\alpha}(x) = 0$,
    $a_j(x)$ will admit $\kappa+1-q$ nontrivial solutions of the form,
    \begin{align}
    (a_j)^i &= a_0 \Gamma(\alpha+1) L_j^{\alpha}(x_i), \qquad j=0, \dots, \kappa-q+1
    \end{align}
    with $x_i$ being the solution of $L_{\kappa+2-q}^{\alpha}(x) = 0$.
\item $L_{\kappa+2-q}^{\alpha}(x) \ne 0$ and $L_{\kappa+1}^{\beta}(x) = 0$, then
    $b_j(x)$ admit $\kappa+1$ nontrivial solutions. However, due to the mixing term, $a_j(x)$ cannot vanish.
    As was shown by the preceding lemma \ref{th.pierwsze}, for $\frac{1}{2}(\beta-\alpha) = 1-q$
    there exists a unique solution to the initial set of equations with $a_j(x)$ and $b_j(x)$ nonvanishing,
    given by
    \begin{align}
    b_j(x_i) &= b_0 \ \Gamma(\beta+1) \ L_j^{\beta}(x_i), \qquad j=0, \dots, \kappa
    \end{align}
    and
    \begin{align}
    a_j(x_i) &= 0, \qquad j=0, \dots, k -1 \\
    a_j(x_i) &= \gamma \ b_0 \ \Gamma(\beta+1) L_{j-k}^{\beta}(x_i), \qquad j=k, \dots, \kappa+k
    \end{align}
    with $x_i$ being the solution of $L_{\kappa+1}^{\beta}(x) = 0$ and
    \begin{equation}
    \gamma = \frac{4 \chi}{\beta^2 - \alpha^2}.
    \end{equation}
\end{itemize}
\end{proof}

\subsection{Theorem 1}

Now, we will use lemma \ref{th.drugie} to get a solution of a set of $m$ coupled recursion relations.
\begin{theorem}
\label{th.trzecie}
Consider a set of $m+1$ cut recursion relations for the generalized Laguerre polynomials with mixing terms, $\chi_i$, as described below,
\begin{align}
a^0_{j-1}(x) &- \big( 2j + \alpha_0 + 1 - x \big) a^0_{j}(x) + (j+1)(j+\alpha_0+1) a^0_{j+1}(x) - \chi_0 a^1_{j+q_0}(x) = 0, \nonumber \\
a^1_{j-1}(x) &- \big( 2j + \alpha_1 + 1 - x \big) a^1_{j}(x) + (j+1)(j+\alpha_1+1) a^1_{j+1}(x) - \chi_1 a^2_{j+q_1}(x) = 0, \nonumber \\
a^2_{j-1}(x) &- \big( 2j + \alpha_2 + 1 - x \big) a^2_{j}(x) + (j+1)(j+\alpha_2+1) a^2_{j+1}(x) - \chi_2 a^3_{j+q_2}(x) = 0, \nonumber \\
&\vdots \nonumber \\
a^{m-1}_{j-1}(x) &- \big( 2j + \alpha_{m-1} + 1 - x \big) a^{m-1}_{j}(x) + \nonumber \\
                &+ (j+1)(j+\alpha_{m-1}+1) a^{m-1}_{j+1}(x) - \chi_{m-1} a^m_{j+q_{m-1}}(x) = 0, \nonumber \\
a^{m}_{j-1}(x) &- \big( 2j + \alpha_{m} + 1 - x \big) a^{m}_{j}(x) + (j+1)(j+\alpha_{m}+1) a^{m}_{j+1}(x) = 0, \nonumber
\end{align}

Then, assuming that there are $\kappa+1$ coefficients $a^m_j(x)$, there will exist $\kappa+1$
nontrivial solutions for $x$ such that $L^{\alpha_m}_{\kappa+1}(x)=0$, provided that for all integers $0<i\le m$ we have, ($k_m \equiv 0$),
\begin{align}
\begin{split}
\frac{1}{2}(\alpha_m-\alpha_{m-i}) &= k_{m-i},  \\
q_{m-i}&=1-k_{m-i}+k_{m-i+1}.
\end{split}
\end{align}
\end{theorem}
\begin{proof}
We start by solving the set of equations for $a^m_j(x)$ using the results of lemma \ref{lem. pierwszy}.
Two situations can be possible, either $x$ is a solution to the equation $L_{\kappa+1}^{\alpha_m}(x) = 0$, or it is not.

\begin{itemize}
\item $L_{\kappa+1}^{\alpha_m}(x) \ne 0$ \\

If $L_{\kappa+1}^{\alpha_m}(x) \ne 0$, lemma \ref{lem. pierwszy} states that all variables $a^m_j(x)$ must vanish. In this case
the equations for the variables $a^{m-1}_j(x)$ become homogeneous. Hence,
we recover the assumptions of the present theorem with $m$ replaced by $m-1$.

\item $L_{\kappa+1}^{\alpha_m}(x) = 0$ \\

Lemma \ref{lem. pierwszy} predicts  $\kappa+1$ nonzero solutions for all variables $a^m_j(x)$.
\end{itemize}

In the following we assume without loss of generality the second possibility, $x$ being
a solution of $L_{\kappa+1}^{\alpha_m}(x) = 0$.
The set of equations for the variables $a^m_j(x)$ and $a^{m-1}_j(x)$ satisfies the assumptions of lemma \ref{th.drugie},
provided that
\begin{align}
\begin{split}
\frac{1}{2}(\alpha_m-\alpha_{m-1}) &= k_{m-1} \textrm{ with } k_{m-1} \textrm{ a positive integer,} \\
q_{m-1} &= 1-k_{m-1}.
\end{split}
\end{align}
According to lemma
\ref{th.drugie} we need $\kappa+2-q_{m-1}$ coefficients $a^{m-1}_j(x)$ for consistency.
Moreover, this theorem gives the solutions for $a^{m-1}_j(x)$ as being proportional to $a^m_j(x)$
and introduces the mixing constant
\begin{equation}
\gamma_{m-1} = \frac{4\chi_{m-1}}{\alpha_m^2 - \alpha_{m-1}^2},
\end{equation}
namely,
\begin{equation}
a^{m-1}_j(x) = a_0^{m} \ \frac{4\chi_{m-1}}{\alpha_m^2 - \alpha_{m-1}^2} \ \Gamma(\alpha_m+1) L_{j-k_{m-1}}^{\alpha_m}(x).
\end{equation}
Consequently, considering the set of equations for $a^{m-2}_j(x)$ we have
\begin{align}
a^{m-2}_{j-1}(x) &- \big( 2j + \alpha_{m-2} + 1 - x \big) a^{m-2}_{j}(x) + \nonumber \\
                &+ (j+1)(j+\alpha_{m-2}+1) a^{m-2}_{j+1}(x) + \chi_{m-2} a^{m-1}_{j+q_{m-2}}(x) = 0.
\end{align}
This can be solved using lemma \ref{th.pierwsze}. We assume that
\begin{equation}
a^{m-2}_j(x) = \gamma_{m-2} \ a_0^m \ \Gamma(\alpha_m+1) a^{m-1}_{j+q_{m-2}}(x)
\end{equation}
and then, provided that,
\begin{align}
\begin{split}
\frac{1}{2}(\alpha_m-\alpha_{m-2}) &= k_{m-2} \textrm{ with } k_{m-2} \textrm{ a positive integer, } \\
q_{m-2} &= 1-k_{m-2} +k_{m-1}.
\end{split}
\end{align}
we find another mixing constant, namely,
\begin{equation}
\gamma_{m-2} =  \frac{4\chi_{m-2}}{\alpha_m^2 - \alpha_{m-2}^2}.
\end{equation}
Again, we need $\kappa+3-q_{m-1}-q_{m-2}$ coefficients $a^{m-2}_j(x)$.
Hence, finally,
\begin{equation}
a^{m-2}_j(x) = a_0^{m} \gamma_{m-1} \gamma_{m-2} \Gamma(\alpha_m+1) L^{\alpha_m}_{j-k_{m-2}}(x),
\end{equation}
which we write as
\begin{equation}
a^{m-2}_j(x) = a_0^{m} \Gamma_{m-1,m-2} \Gamma(\alpha_m+1) L^{\alpha_m}_{j-k_{m-2}}(x),
\end{equation}
with
\begin{equation}
\Gamma_{x,y} = \prod_{t=y}^x \gamma_t.
\end{equation}
One can repeat these steps $m-2$ times and find the solution
for all $a^p_j(x)$.\\

Summarizing, for all $\kappa+1$ roots of the equation $L_{\kappa+1}^{\alpha_m}(x) = 0$, denoted by $x_i$, we have a nontrivial solution of the form,
\begin{align}
&a^m_j(x_i) =  a_0^m \ \Gamma(\alpha_m+1) L^{\alpha_m}_j(x_i), \nonumber \\
            & \qquad\qquad\qquad \qquad\qquad\qquad \textrm{with } \qquad 0 \le j \le \kappa \nonumber \\
&a^{m-1}_j(x_i) = a_0^m \ \Gamma_{m-1,m-1} \ \Gamma(\alpha_m+1) L_{j-k_{m-1}}^{\alpha_m}(x_i), \nonumber \\
            & \qquad\qquad\qquad \qquad\qquad\qquad \textrm{with } \qquad k_{m-1} \le j \le \kappa + k_{m-1} \nonumber \\
&a^{m-2}_j(x_i) = a_0^m \ \Gamma_{m-1,m-2} \ \Gamma(\alpha_m+1) L_{j-k_{m-1}-k_{m-2}}^{\alpha_m}(x_i), \nonumber \\
            & \qquad\qquad\qquad \qquad\qquad\qquad \textrm{with } \qquad k_{m-1}+k_{m-2} \le j \le \kappa+k_{m-1}+k_{m-2} \nonumber \\
& \qquad\qquad\qquad \qquad\qquad\qquad \vdots \nonumber \\
&a^{1}_j(x_i) = a_0^m \ \Gamma_{m-1,1}  \ \Gamma(\alpha_m+1) L_{j-\sum_{t=1}^{m-1} k_{t}}^{\alpha_m}(x_i), \nonumber \\
            & \qquad\qquad\qquad \qquad\qquad\qquad \textrm{with } \qquad \Big(\sum_{t=1}^{m-1} k_{t}\Big) \le j \le \kappa+\Big(\sum_{t=1}^{m-1} k_t \Big) \nonumber \\
&a^{0}_j(x_i) = a_0^m \ \Gamma_{m-1,0}  \ \Gamma(\alpha_m+1) L_{j-\sum_{t=0}^{m-1} k_{t}}^{\alpha_m}(x_i), \nonumber \\
            & \qquad\qquad\qquad \qquad\qquad\qquad \textrm{with } \qquad \Big(\sum_{t=0}^{m-1} k_{t}\Big) \le j \le \kappa+ \Big(\sum_{t=0}^{m-1} k_t \Big). \nonumber
\end{align}
\end{proof}

\subsection{Application of theorem 1 to the recurrence in the bosonic sector}

We apply theorem \ref{th.trzecie} to the recurrence relation eq.\eqref{eq.su3.recurrence.rel}.

\subsubsection{Even solutions}

In the case of even solutions we get the identification
\begin{align}
\begin{array}{cl}
\alpha_t &= \ 6t+3, \\
\chi_t &= \ -\frac{3}{4} (t+1)(2t+1), \\
k_t &= \ \frac{1}{2}(\alpha_t - \alpha_{t-1}) = 3,\\
q_t &= \ -2 = 1 - k_t.
\end{array}
\end{align}
Therefore, we obtain
\begin{equation}
\gamma_t = -\frac{1}{3} \frac{(t+1)(2t+1)}{(2\kappa+1)^2-(2t+1)^2},
\end{equation}
and
\begin{equation}
\Gamma^{even}_{\kappa-1,t} = \prod_{p=t}^{\kappa-1} \gamma_{p} = (-24)^{t-\kappa}\frac{(\kappa+t)!}{(\kappa-t)!(2t)!}.
\end{equation}

\subsubsection{Odd solutions}

In the case of odd solutions, we can identify as follows
\begin{align}
\begin{array}{cl}
\alpha_t &= \ 6t+6,  \\
\chi_t &= \ -\frac{3}{4} (t+1)(2t+3),  \\
k_t &= \ \frac{1}{2}(\alpha_t - \alpha_{t-1}) = 3,  \\
q_t &= \ -2 = 1 - k_t,
\end{array}
\end{align}
and obtain
\begin{equation}
\gamma_t = -\frac{1}{6} \frac{(t+1)(t+\frac{3}{2})}{(\kappa+1)^2-(t+1)^2}.
\end{equation}
Thus,
\begin{equation}
\Gamma^{odd}_{\kappa-1,p} = \prod_{t=p}^{\kappa-1} \gamma_{t} = (-24)^{t-\kappa}\frac{(\kappa+t+1)!}{(\kappa-t)!(2t+1)!}.
\end{equation}

\subsection{Theorem 2}

Eventually, we need the solutions to a twofold coupled Laguerre recursion relations. We present it in a form of the
following theorem.

\begin{theorem}
\label{th.szoste}
Consider a set of $m+1$ cut recursion relations for the generalized Laguerre polynomials with mixing terms as described below,
\begin{align}
\textbf{S}^{\alpha_0}_j \cdot  \textbf{a}^0_{j}(x) &- \chi_0 \ a^1_{j+q_0}(x) - \mu_0 \ a^2_{j+p_0}(x) = 0, \nonumber \\
\textbf{S}^{\alpha_1}_j \cdot  \textbf{a}^1_{j}(x) &- \chi_1 \ a^2_{j+q_1}(x) - \mu_1 \ a^3_{j+p_1}(x) = 0, \nonumber \\
\textbf{S}^{\alpha_2}_j \cdot  \textbf{a}^2_{j}(x) &- \chi_2 \ a^3_{j+q_2}(x) - \mu_2 \ a^4_{j+p_2}(x) = 0, \nonumber \\
%S^{\alpha_3}_j \cdot  a^3_{j}(x) &- \chi_3 \ a^4_{j+q_3}(x) - \mu_3 \ a^5_{j+p_3}(x) = 0, \nonumber \\
&\vdots \nonumber \\
\textbf{S}^{\alpha_{m-3}}_j \cdot  \textbf{a}^{m-3}_{j}(x) &- \chi_{m-3} \ a^{m-2}_{j+q_{m-3}}(x) - \mu_{m-3} \ a^{m-1}_{j+p_{m-1}}(x) = 0, \nonumber \\
\textbf{S}^{\alpha_{m-2}}_j \cdot  \textbf{a}^{m-2}_{j}(x) &- \chi_{m-2} \ a^{m-1}_{j+q_{m-2}}(x) - \mu_{m-2} \ a^m_{j+p_{m-2}}(x) = 0, \nonumber \\
\textbf{S}^{\alpha_{m-1}}_j \cdot  \textbf{a}^{m-1}_{j}(x) &- \chi_{m-1} \ a^m_{j+q_{m-1}}(x) = 0, \nonumber \\
\textbf{S}^{\alpha_m}_j \cdot  \textbf{a}^m_{j}(x) &= 0,
\end{align}

If $x$ is one of the $\kappa+1$ roots of the equation $L_{\kappa+1}^{\alpha_m}(x)=0$ then, there exist one nontrivial solution
with variables $a^{p}_t(x)$ partially nonzero (for each $p$ at least for one value of the index $t$ $a^{p}_t(x) \ne 0$)
, provided that for all integers $0 < i \le m$ we have, ($k_t \equiv 0$, for $t \ge m$)
\begin{align}
\begin{split}
\frac{1}{2}(\alpha_m-\alpha_{m-i}) &= k_{m-i},  \\
p_{m-i} &= 1 - k_{m-i} + k_{m-i+2},  \\
q_{m-i} &= 1 - k_{m-i} + k_{m-i+1}.
\end{split}
\end{align}
\end{theorem}
\begin{proof}
The prove goes in a similar manner as the proof of theorem \ref{th.trzecie}. We start by solving the set of equations for $a^m_j$
using the results of lemma \ref{lem. pierwszy}. By assumption we have that
$L_{\kappa+1}^{\alpha_m}(x) = 0$.
The set of equations for $a^m_j(x)$ and $a^{m-1}_j(x)$ satisfies the assumptions of lemma \ref{th.drugie},
provided that
\begin{align}
\begin{split}
\frac{1}{2}(\alpha_m-\alpha_{m-1}) &= k_{m-1} \textrm{ with } k_{m-1} \textrm{ a positive integer,}  \\
q_{m-1} &= 1-k_{m-1}.
\end{split}
\end{align}
Therefore,
lemma \ref{th.drugie} gives the solutions for $a^{m-1}_j(x)$ as being proportional to $a^m_j(x)$
and introduces a mixing constant,
\begin{equation}
\gamma_{m-1} = \frac{4\chi_{m-1}}{\alpha_m^2 - \alpha_{m-1}^2}.
\end{equation}
Consequently, considering the set of equations for $a^{m-2}_j(x)$ we have,
\begin{align}
\textbf{S}^{\alpha_{m-2}}_j &\cdot  \textbf{a}^{m-2}_{j}(x) - \chi_{m-2} a^{m-1}_{j+q_{m-2}}(x) - \mu_{m-2} a^m_{j+p_{m-2}}(x) = 0,
\end{align}
where
\begin{align}
a^{m-1}_j(x) &= a_0^m \ \gamma_{m-1} \Gamma(\alpha_m+1) L_{j-k_{m-1}}^{\alpha_m}(x),  \\
a^{m}_j(x) &= a_0^m \ \Gamma(\alpha_m+1) L_{j}^{\alpha_m}(x).
\end{align}
Thus, we have
\begin{align}
\textbf{S}^{\alpha_{m-2}}_j &\cdot \textbf{a}^{m-2}_{j}(x) - \chi_{m-2}\gamma_{m-1} a_0^m \Gamma(\alpha_m+1) L_{j-k_{m-1}+q_{m-2}}^{\alpha_m}(x) \nonumber \\
&- \mu_{m-2}  a_0^m \Gamma(\alpha_m+1) L_{j+p_{m-2}}^{\alpha_m}(x) = 0.
\end{align}
Provided that,
\begin{align}
\begin{split}
\frac{1}{2}(\alpha_m-\alpha_{m-2}) &= k_{m-2} \textrm{ with } k_{m-2} \textrm{ a positive integer,} \\
p_{m-2} &= 1-k_{m-2}, \\
q_{m-2} &= 1 -k_{m-2} + k_{m-1},
\end{split}
\end{align}
we can use lemma \ref{th.pierwsze} and obtain the solution for $a^{m-2}_{j}(x)$ as being proportional to $a^m_j(x)$ with a mixing constant given by,
\begin{equation}
\gamma_{m-2} = \frac{4(\chi_{m-2}\gamma_{m-1} + \mu_{m-2})}{\alpha_m^2 - \alpha_{m-2}^2}.
\end{equation}
Next, the set of equations for $a^{m-3}_j(x)$ reads,
\begin{align}
\textbf{S}^{\alpha_{m-3}}_j &\cdot  \textbf{a}^{m-3}_{j}(x) - \chi_{m-3} a^{m-2}_{j+q_{m-3}}(x) - \mu_{m-3} a^{m-1}_{j+p_{m-3}}(x) = 0,
\end{align}
where $a^{m-1}_j(x)$ and $a^{m-2}_j(x)$ are given now by
\begin{align}
a^{m-2}_j(x) &= a_0^m \ \gamma_{m-2} \ \Gamma(\alpha_m+1) L_{j-k_{m-2}}^{\alpha_m}(x), \\
a^{m-1}_j(x) &= a_0^m \ \gamma_{m-1} \ \Gamma(\alpha_m+1) L_{j-k_{m-1}}^{\alpha_m}(x).
%a^{m}_j(x) &= a_0^m \ \Gamma(\alpha_m+1) L_{j}^{\alpha_m}(x). \nonumber
\end{align}
Thus, we have
\begin{align}
\textbf{S}^{\alpha_{m-3}}_j \cdot \textbf{a}^{m-3}_{j}(x) &- a_0^m \chi_{m-3}\gamma_{m-2} \Gamma(\alpha_m+1) L_{j-k_{m-2}+q_{m-3}}^{\alpha_m}(x) \nonumber \\
&- a_0^m \mu_{m-3}\gamma_{m-1} \Gamma(\alpha_m+1) L_{j-k_{m-1}+p_{m-3}}^{\alpha_m}(x) = 0,
\end{align}
Again, provided that
\begin{align}
\begin{split}
\frac{1}{2}(\alpha_m-\alpha_{m-3}) &= k_{m-3} \textrm{ with } k_{m-3} \textrm{ a positive integer,} \\
p_{m-3} &= 1 - k_{m-3} + k_{m-1}, \\
q_{m-3} &= 1-k_{m-3}+k_{m-2},
\end{split}
\end{align}
we can use lemma \ref{th.pierwsze} and obtain the solution for $a^{m-3}_{j}(x)$ as
being proportional to $a^m_j(x)$ with a mixing constant given by,
\begin{equation}
\gamma_{m-3} = \frac{4(\chi_{m-3}\gamma_{m-2} + \mu_{m-3}\gamma_{m-1})}{\alpha_m^2 - \alpha_{m-3}^2}.
\end{equation}
One can repeat these steps $m-3$ times and find the solution
for all $a^p_j$.\\

Summarizing, for all roots $x_i$ of the equation $L_{\kappa+1}^{\alpha_m}(x) = 0$ we have a nontrivial solution of the form,
\begin{align}
&a^m_j(x_i) =  a^m_0 \ \Gamma(\alpha_m+1) \ L^{\alpha_m}_j(x_i), \nonumber \\
            & \qquad\qquad\qquad \qquad\qquad\qquad \textrm{with } \qquad 0 \le j \le \kappa \nonumber \\
&a^{m-1}_j(x_i) = a^m_0 \ \Gamma_{m-1,m-1} \ \Gamma(\alpha_m+1)  L_{j-k_{m-1}}^{\alpha_m}(x_i), \nonumber \\
            & \qquad\qquad\qquad \qquad\qquad\qquad \textrm{with } \qquad k_{m-1} \le j \le \kappa+k_{m-1} \nonumber \\
&a^{m-2}_j(x_i) = a^m_0 \ \Gamma_{m-1,m-2} \ \Gamma(\alpha_m+1)  L_{j-k_{m-1}-k_{m-2}}^{\alpha_m}(x_i), \nonumber \\
            & \qquad\qquad\qquad \qquad\qquad\qquad \textrm{with } \qquad k_{m-1}+k_{m-2} \le j \le \kappa+k_{m-1}+k_{m-2} \nonumber
\end{align}
\begin{align}
&\qquad\qquad\qquad \qquad\qquad\qquad \vdots \nonumber \\
&a^{1}_j(x_i) = a^m_0 \ \Gamma_{m-1,1} \ \Gamma(\alpha_m+1)  L_{j-\sum_{t=1}^{m-1} k_{t}}^{\alpha_m}(x_i), \nonumber \\
            & \qquad\qquad\qquad \qquad\qquad\qquad \textrm{with } \qquad \Big( \sum_{t=1}^{m-1} k_{t} \Big) \le j \le \kappa+\Big(\sum_{t=1}^{m-2}k_t\Big) \nonumber \\
&a^{0}_j(x_i) = a^m_0 \ \Gamma_{m-1,0} \ \Gamma(\alpha_m+1)  L_{j-\sum_{t=0}^{m-1} k_{t}}^{\alpha_m}(x_i),  \nonumber \\
            & \qquad\qquad\qquad \qquad\qquad\qquad \textrm{with } \qquad \Big( \sum_{t=0}^{m-1} k_{t} \Big) \le j \le \kappa+\Big(\sum_{t=0}^{m-1}k_t\Big) \nonumber
\end{align}
where we have introduced by force a similar notation as used in the previous theorem, namely,
\begin{equation}
\Gamma_{x,y} = \gamma_y,
\end{equation}
with $\gamma_y$ defined recursively by
\begin{equation}
\gamma_y = 4\frac{\chi_y \gamma_{y+1} + \mu_y \gamma_{y+2}}{\alpha_m^2 - \alpha_{y}^2},
\end{equation}
and $\gamma_{y \ge m} \equiv 1$.
\end{proof}

\subsection{Application of theorem 2 to the recurrence in the sector $n_F=1$}

In this subsection we use the theorem \ref{th.szoste} to find solutions of the recurrence relation in the sector with $n_F=1$.

\subsubsection{Even solutions}

In order to obtain solutions of eqs.\eqref{eq. f1_recurrence relation even}, we identify as follows
\begin{align}
\begin{array}{cl}
\alpha_t &= \ \left\{ \begin{array}{ll}
3t+4 ,& t \textrm{ odd},\\
3t+5 , & t \textrm{ even},
\end{array} \right. \\
\chi_t &= \ \left\{ \begin{array}{ll}
-\frac{1}{2}(t+1), & t \textrm{ odd},\\
-\frac{3}{2}(t+1), & t \textrm{ even},
\end{array} \right.  \\
\mu_t &= \ -\frac{3}{8}(t+1)(t+2),
\end{array}
\end{align}
There is neither closed expression for $\gamma_y$ nor for $\Gamma_{x,y}$.

\subsubsection{Odd solutions}

In a similar manner, solutions of eqs.\eqref{eq. f1_recurrence relation odd}, can be obtained after the following identification,
\begin{align}
\begin{array}{cl}
\alpha_t &= \ \left\{ \begin{array}{ll}
3t+4, & t \textrm{ even},\\
3t+5 ,& t \textrm{ odd},
\end{array} \right. \\
\chi_t &= \ \left\{ \begin{array}{ll}
-\frac{t}{2}, & t \textrm{ even},\\
-\frac{3}{2}(t+1), & t \textrm{ odd},
\end{array} \right. \\
\mu_t &= \ -\frac{3}{8}(t+1)(t+2).
\end{array}
\end{align}
Similarly, in this case closed expressions for $\gamma_y$ and for $\Gamma_{x,y}$ do not exist.

\subsection{Corollary}

\begin{corollary}
\label{th.piate}
The set of all possible values of the $x$ parameter for which there exist
a nontrivial solution of the cut set of generalized Laguerre recursion relations
does not depend on the precise form of the mixing coefficients $\chi$ and $\mu$ as long as $\chi$ and $\mu$ do not depend on $j$.
\end{corollary}
\begin{proof}
The set of $m+1$ cut generalized Laguerre recursion relations of the form presented,
in theorems \ref{th.trzecie} or \ref{th.szoste}, have a specific structure, namely,
the $i$-th recursion relation contains mixing terms proportional only
to the variables of the $j$-th recursion relation where $j>i$. Let us assume that $\forall j>i$
all the amplitudes described by the $j$-th recursion relation vanish.
Then, the $i$-th recursion relation has no mixing terms and can be solved by the lemma \ref{lem. pierwszy}.
Hence, it follows from the lemma \ref{lem. pierwszy} that
the possible values of the $x$ parameter are given by the equation
\begin{equation}
L_{\kappa_i+1}^{\alpha_i}(x) = 0,
\end{equation}
with some appropriate $\kappa_i$.
Let us now assume that $x$ is not a solution of the above equation. Then, the amplitudes
described by the $i$-th recursion relation must all vanish. In this case, the $(i-1)$-th recursion
relation has no mixing terms and can be solved by the lemma \ref{lem. pierwszy}. The set of possible
values of the $x$ parameter is therefore given by the sum of zeros of the equations,
\begin{align}
L_{\kappa_i+1}^{\alpha_i}(x) &= 0, \\
L_{\kappa_{i-1}+1}^{\alpha_{i-1}}(x) &= 0.
\end{align}
By induction, we can conclude that the set of all possible values of the $x$ parameter for which there exist
a nontrivial solution of the cut set of $m+1$ generalized Laguerre recursion relations is given by the equation
\begin{equation}
\Bigg( \prod_{p=0}^m L_{\kappa_p+1}^{\alpha_p}(x) \Bigg) = 0.
\end{equation}
We complete the proof by noting that the cut set of $m+1$ generalized Laguerre recursion relations
by assumption contains one recursion relation without mixing terms which can be solved
with the use of lemma \ref{lem. pierwszy}.
\end{proof}


\small

\begin{thebibliography}{99}
\bibitem{luscher1} M. L\"uscher, 'Some analytic results concerning the mass spectrum of Yang-Mills gauge thoeries on a torus', Nucl. Phys. B 219 (1983) 233-261
\bibitem{luscher2} M. L\"uscher, G. M\"unster, 'Weak-coupling expansion of the low-lying energy values in the $SU(2)$ gauge theory on a torus', Nucl. Phys. B 232 (1984) 445-472
\bibitem{baal1} J. Koller, P. van Baal, 'A non-perturbative analysis in finite volume gauge theory', Nucl. Phys. B 302 (1988) 1-64
\bibitem{baal2} P. van Baal, 'Gauge theory in a finite volume', Acta Phys. Pol. B 20 (1989) 295
\bibitem{hoppe} J. Hoppe, 'Quantum theory of a massless relativistic surface and a two dimensional bound state problem', PhD thesis MIT, 1982, unpublished (scanned version avaible at http://www.aei-potsdam.mpg.de/ hoppe)
\bibitem{dewit} B. de Wit, J. Hoppe, H. Nicolai, 'On the quantum mechanics of supermembranes', Nucl. Phys. B 305 (1988) 545
\bibitem{wosiek1} J. Wosiek, 'Spectra of supersymmetric Yang-Mills quantum mechanics', Nucl. Phys. B 644 (2002) 85-112
%\bibitem{wosiek2} J. Wosiek, 'Supersymmetric Yang-Mills quantum mechanics', ...
\bibitem{wosiek3} J. Wosiek, 'Supersymmetric Yang-Mills quantum mechanics in various dimensions', PoS LAT 2005 (2006) 273
\bibitem{wosiek4} M. Campostrini, J. Wosiek, 'Exact Witten index in D=2 supersymmetric Yang-Mills quantum mechanics', Phys. Lett. B 550 (2002) 121-127
\bibitem{wosiek5} M. Campostrini, J. Wosiek, 'High precision study of the structure of D=4 supersymmetric Yang-Mills quantum mechanics', Nucl. Phys. B 703 (2004) 454-498
%\bibitem{kotanski1} J. Kotanski, 'Virial theorem for four-dimensional supersymmetric Yang-Mills quantum mechanics with SU(2) gauge group', Acta Phys. Pol ...
\bibitem{kotanski2} J. Kotanski, 'Energy spectrum and wave-functions of four-dimensional supersymmetric Yang-Mills quantum mechanics for very high cut-offs', Acta Phys.Pol. B 37 (2006) 2813-2838
\bibitem{doktorat_macka} M. Trzetrzelewski, 'Supersymmetric Yang-Mills quantum mechanics with arbitrary number of colors', Ph.D. thesis, Jagiellonian University
\bibitem{korcyl_phd} P. Korcyl, Ph.D. thesis in preparation
%\bibitem{korcyl3} P. Korcyl, 'Study of a supersymmetric quantum systems with discrete spectrum', to be published
%\bibitem{korcyl4} P. Korcyl, 'Bosonic and fermionic large-$N$ solutions in $D=2$ supersymmetric Yang-Mills quantum mechanics', to be published
\bibitem{claudson} M. Claudson, M.B. Halpern, 'Ground state wave functions', Nucl. Phys. B 250 (1985) 689-715
\bibitem{samuel} S. Samuel, 'Solutions of extended supersymmetric matrix models for arbitrary gauge groups', Phys. Lett B 411 (1997) 268-273
\bibitem{maciek3} M. Trzetrzelewski, 'Reduction of $SU(N)$ loop tensors to trees', J. Math. Phys. 46 (2005) 103512
\bibitem{korcyl0} P. Korcyl, 'Eigenvalues and eigenvectors of the $d$ dimensional Laplace operator in a cut Fock basis', [arXiv:1008.0778]
\bibitem{korcyl1} P. Korcyl, 'Recursive approach to supersymmetric quantum mechanics for arbitrary fermion occupation number',  Acta Phys. Pol. B 41 (2010) 795, [arXiv:0912.5265]
\bibitem{maciek2} M. Trzetrzelewski, J. Wosiek, 'Quantum systems in a cut Fock space', Acta Phys.Polon. B35 (2004) 1615-1624, [hep-th/0308007]
\bibitem{maciek1} M. Trzetrzelewski, 'Quantum mechanics in a cut Fock space', Acta Phys. Polon. B 35 (2004) 2393-2416, [hep-th/0407059]
\bibitem{korcyl_msc} P. Korcyl, 'Classical trajectories and quantum supersymmetry', Phys. Rev. D 74 (2006) 115012,
%\bibitem{witten1} E. Witten, 'Dynamical breaking of supersymmetry', Nucl. Phys. B 185 (1981) 513-554
%\bibitem{witten2} E. Witten, 'Constraints on supersymmetry breaking', Nucl. Phys. B 202 (1982) 253-316
\bibitem{abramowitz} M. Abramowitz, I.A.Stegun, 'Handbook of Mathematical Functions with Formulas, Graphs, and Mathematical Tables', Dover Publications, New York, 1968
%\bibitem{cooper} F. Cooper, A. Khare, U. Sukhatme, 'Supersymmetry and Quantum Mechanics', Phys.Rept. 251 (1995) 267-385
%\bibitem{praca_magisterska} P. Korcyl, 'Classical trajectories and quantum supersymmetry', Phys. Rev. D 74 (2006) 115012
%\bibitem{korcyl2} P. Korcyl, 'Eigenvalues and eigenvectors of the $d$ dimensional Laplace-Beltrami operator in a cut Fock basis', to be published
%\bibitem{yi} P. Yi, 'Witten Index and Threshold Bound States of D-Branes', Nucl. Phys. B505 (1997) 307, [hep-th/9704098]
%\bibitem{sethi_stern} S. Sethi and M. Stern, 'D-Brane Bound State Redux', Commun.Math.Phys. 194 (1998) 675-705, [hep-th/9705046]
%\bibitem{kac_smilga} V.G. Kac and A.V. Smilga, 'Normalized Vacuum States in N = 4 Supersymmetric Yang-Mills Quantum Mechanics with any Gauge Group', Nucl. Phys. B571 (2000) 515, [hep-th/9908096]
%\bibitem{staudacher} M. Staudacher, 'Bulk Witten Indices and the Number of Normalizable Ground States in Supersymmetric Quantum Mechanics of Orthogonal,Symplectic and Exceptional Groups', Phys.Lett. B488 (2000) 194, [hep-th/0006234]
\end{thebibliography}

\normalsize
\end{document}